\documentclass[11pt,reqno,a4paper]{amsart}

\usepackage{amsmath,amsthm,amssymb,amsfonts} 
\usepackage{graphicx,xcolor} 
\usepackage{bm} 
\usepackage{enumitem} 
\usepackage{lmodern} 
\usepackage{microtype} 
\usepackage[english]{babel} 

\definecolor{darkblue}{rgb}{0,0,.8}

\usepackage[nobreak]{cite} 

\usepackage{tikz} %

\usetikzlibrary{arrows,decorations,calc}
\usetikzlibrary{decorations.pathmorphing,patterns,decorations.pathreplacing,decorations.markings}

\usepackage[hypertexnames=false]{hyperref} 
\hypersetup{
    colorlinks,
    citecolor=red,
    filecolor=black,
    linkcolor=darkblue,
    urlcolor=black
}

\usepackage[noabbrev,capitalize]{cleveref} 

\numberwithin{equation}{section} 

\newtheorem{theorem}{Theorem}[section]
\newtheorem{lemma}[theorem]{Lemma}
\newtheorem{proposition}[theorem]{Proposition}
\newtheorem{corollary}[theorem]{Corollary}

\newtheorem{example}[theorem]{Example}

\tikzset{
  on each segment/.style={
    decorate,
    decoration={
      show path construction,
      moveto code={},
      lineto code={
        \path [#1]
        (\tikzinputsegmentfirst) -- (\tikzinputsegmentlast);
      },
      curveto code={
        \path [#1] (\tikzinputsegmentfirst)
        .. controls
        (\tikzinputsegmentsupporta) and (\tikzinputsegmentsupportb)
        ..
        (\tikzinputsegmentlast);
      },
      closepath code={
        \path [#1]
        (\tikzinputsegmentfirst) -- (\tikzinputsegmentlast);
      },
    },
  },
  mid arrow/.style={postaction={decorate,decoration={
        markings,
        mark=at position .625 with {\arrow[#1]{stealth}}
      }}},
}

\newcommand\yaxis{90}
\newcommand\xaxis{0}

\newcommand{\drawvertex}[3]
{
  \fill[shift={(\xaxis:#2)},shift={(\yaxis:#3)}] (0,0) circle (1.5pt);
  \ifnum#1=1
   \draw[shift={(\xaxis:#2)},shift={(\yaxis:#3)},postaction={on each segment={mid arrow}}] (-.5,0)--(0,0)--(.5,0);
  \draw[shift={(\xaxis:#2)},shift={(\yaxis:#3)},postaction={on each segment={mid arrow}}] (0,-.5)--(0,0)--(0,.5);
  \fi
   \ifnum#1=2
      \draw[shift={(\xaxis:#2)},shift={(\yaxis:#3)},postaction={on each segment={mid arrow}}] (.5,0)--(0,0)--(-.5,0);
  \draw[shift={(\xaxis:#2)},shift={(\yaxis:#3)},postaction={on each segment={mid arrow}}] (0,.5)--(0,0)--(0,-.5); 
  \fi
\ifnum#1=3
      \draw[shift={(\xaxis:#2)},shift={(\yaxis:#3)},postaction={on each segment={mid arrow}}] (.5,0)--(0,0)--(-.5,0);
  \draw[shift={(\xaxis:#2)},shift={(\yaxis:#3)},postaction={on each segment={mid arrow}}] (0,-.5)--(0,0)--(0,.5);   \fi
\ifnum#1=4
      \draw[shift={(\xaxis:#2)},shift={(\yaxis:#3)},postaction={on each segment={mid arrow}}] (-.5,0)--(0,0)--(.5,0);
  \draw[shift={(\xaxis:#2)},shift={(\yaxis:#3)},postaction={on each segment={mid arrow}}] (0,.5)--(0,0)--(0,-.5);  
  \fi
\ifnum#1=5
      \draw[shift={(\xaxis:#2)},shift={(\yaxis:#3)},postaction={on each segment={mid arrow}}] (-.5,0)--(0,0)--(0,.5);
  \draw[shift={(\xaxis:#2)},shift={(\yaxis:#3)},postaction={on each segment={mid arrow}}] (.5,0)--(0,0)--(0,-.5);  
  \fi
\ifnum#1=6
      \draw[shift={(\xaxis:#2)},shift={(\yaxis:#3)},postaction={on each segment={mid arrow}}] (0,.5)--(0,0)--(-.5,0);
  \draw[shift={(\xaxis:#2)},shift={(\yaxis:#3)},postaction={on each segment={mid arrow}}] (0,-.5)--(0,0)--(.5,0) ; 
  \fi
}

\newcommand{\drawcornervertex}[3]
{
  \ifnum#1=0
  \fi
  \ifnum#1=1
    \fill[shift={(\xaxis:#2)},shift={(\yaxis:#3)}] (0,0) circle (1.5pt);
    \draw[shift={(\xaxis:#2)},shift={(\yaxis:#3)},postaction={on each segment={mid arrow}}] (0,-.5) -- (0,0) -- (.5,0);
  \fi
  \ifnum#1=2
    \fill[shift={(\xaxis:#2)},shift={(\yaxis:#3)}] (0,0) circle (1.5pt);
    \draw[shift={(\xaxis:#2)},shift={(\yaxis:#3)}, postaction={on each segment={mid arrow}}] (.5,0) -- (0,0) -- (0,-.5);
  \fi
  \ifnum#1=3
    \fill[shift={(\xaxis:#2)},shift={(\yaxis:#3)}] (0,0) circle (1.5pt);
    \draw[shift={(\xaxis:#2)},shift={(\yaxis:#3)}, postaction={on each segment={mid arrow}}] (.5,0) -- (0,0);
    \draw[shift={(\xaxis:#2)},shift={(\yaxis:#3)}, postaction={on each segment={mid arrow}}] (0,-.5) -- (0,0);
  \fi
  \ifnum#1=4
    \fill[shift={(\xaxis:#2)},shift={(\yaxis:#3)}] (0,0) circle (1.5pt);
    \draw[shift={(\xaxis:#2)},shift={(\yaxis:#3)}, postaction={on each segment={mid arrow}}] (0,0) -- (.5,0);
    \draw[shift={(\xaxis:#2)},shift={(\yaxis:#3)}, postaction={on each segment={mid arrow}}] (0,0) -- (0,-.5);
  \fi
}

\newcounter{x}
\newcounter{y}

\newcommand{\confsixvertex}[3]{
 \begin{scope}[shift={(\xaxis:#2)},shift={(\yaxis:#3)}]
 \setcounter{y}{0}
  \foreach \a in {#1}
  {
    \setcounter{x}{0}
    \foreach \b in \a 
    {
    \drawcornervertex{\b}{.5*\value{x}}{.5*\value{y}}
    \addtocounter{x}{1}
    }
    
    \fill (.5*\value{y},-.5) circle (1.5pt);
    \fill (.5*\value{x},.5*\value{y}) circle (1.5pt);

    \addtocounter{y}{1}
  }
  \end{scope}
}

\newcommand{\ii}{\mathrm{i}}
\newcommand{\ee}{\mathrm{e}}
\newcommand{\diff}{\text{d}}

\newcommand{\ket}[1]{|{#1}\rangle}
\newcommand{\bra}[1]{\langle {#1}|}
\newcommand{\uu}{\uparrow}
\newcommand{\dd}{\downarrow}

\newcommand{\RR}{\mathbb{R}}
\newcommand{\II}{\mathbb{I}}
\newcommand{\KK}{\mathbb{K}}
\newcommand{\ZZ}{\mathbb{Z}}
\newcommand{\YY}{\mathbb{Y}}
\newcommand{\MM}{\mathbb{M}}
\newcommand{\WW}{\mathbb{W}}

\begin{document}

\title[Open XXZ chain and TSASMs] {\large The open XXZ chain at $\bm{\Delta =-1/2}$ and totally-symmetric alternating sign matrices}

\author{Jean Li\'enardy}
\address
{\'Ecole Royale Militaire\\ D\'epartement de Math\'ematiques \\ Avenue de la Renaissance 30, B-1000 Bruxelles \\ Belgium}
\email{jean.lienardy@mil.be} 

\author{Christian Walmsley Hagendorf}
\thanks{The second author acknowledges support from the F\'ed\'eration Wallonie-Bruxelles (FWB) through the ARC project 24/29-145 ``Emergent Motifs in Interconnected Systems (EMOTIONS)''}
\address
{Institut de Recherche en Math\'ematique et Physique \\ Universit\'e catholique de Louvain\\
Chemin du Cyclotron 2\\
B-1348 Louvain-la-Neuve\\ Belgium}
\email{christian.hagendorf@uclouvain.be} 

\begin{abstract}
The open XXZ spin chain with the anisotropy parameter $\Delta=-\frac12$, diagonal boundary fields that depend on a parameter $x$, and finite length $N$ is studied. In a natural normalisation, the components of its ground-state vector are polynomials in $x$ with integer coefficients. It is shown that their sum is given by a generating function for the weighted enumeration of totally-symmetric alternating sign matrices with weights depending on $x$.
\end{abstract}

\maketitle 
\tableofcontents 

\section{Introduction}
The ground-state eigenvectors of the XXZ spin chain with certain boundary conditions and the anisotropy parameter $\Delta=-\frac12$ exhibit remarkable connections with enumerative combinatorics \cite{stroganov:01,razumov:01,razumov:01_2,batchelor:01}. An example is provided by the ground-state eigenvectors of the spin-chain Hamiltonian with periodic boundary conditions and $N=2n+1$ sites, where $n\geqslant 1$ is an integer. Razumov and Stroganov investigated these eigenvectors for small values of $n$. Based on explicit calculations, they conjectured that, if their smallest component is normalised to one, then all other components are integers and the largest component is equal to the number of alternating sign matrices (ASMs) of order $n$ \cite{bressoudbook}. This and many similar conjectures were subsequently proven through an exact construction and characterisation of the ground-state eigenvectors in terms of multiple contour integrals and the quantum Knizhnik-Zamodlodchikov equations associated with the six-vertex model \cite{difrancesco:06,razumov:07,cantini:12_1,morin:20}.

In this article, we consider the XXZ spin chain with open boundary conditions and diagonal boundary fields, continuing our investigations started in \cite{hagendorf:21,hagendorf:22}. For a chain of length $N\geqslant 1$, the spin-chain Hamiltonian is an operator on $(\mathbb C^2)^{\otimes N}$, given by
\begin{subequations}
\label{eqn:XXZHamiltonian}
\begin{equation}
  H = -\frac12 \sum_{i=1}^{N-1}\left( \sigma_i^x\sigma_{i+1}^x + \sigma_i^y\sigma_{i+1}^y + \Delta \sigma_i^z\sigma_{i+1}^z\right) + p\sigma_{1}^z+ p'\sigma_N^z.
\end{equation}
Here, $\sigma^x_i,\sigma^y_i,\sigma^z_i$ are the usual Pauli matrices $\sigma^x,\sigma^y,\sigma^z$ acting on the $i$-th tensor factor of $(\mathbb C^2)^{\otimes N}$. Moreover, the Hamiltonian depends on the anisotropy parameter $\Delta$, and the boundary fields $p,p'$. We focus on the case where these three parameters take the values
\begin{equation}
  \label{eqn:XXZParams}
  \Delta =-\frac12, \quad p = \frac12\left(\frac12-x\right),\quad p' = \frac12\left(\frac12-\frac1{x}\right),
\end{equation}
\end{subequations}
where $x$ is an arbitrary non-zero complex number. Special values include $x=1$, where the spin-chain Hamiltonian is supersymmetric \cite{hagendorf:17}, or $x=\ee^{-\ii \pi/3}$, where one obtains the quantum-group invariant Pasquier-Hamiltonian for $\Delta=-1/2$ \cite{pasquier:90}.

The choice of parameters \eqref{eqn:XXZParams} is motivated by the existence of a remarkably simple eigenvalue of the Hamiltonian \cite{degier:04,nichols:05,hagendorf:21},
 \begin{equation}
  E = -\frac{3N-1}{4} - \frac{(1-x)^2}{2x}.
  \label{eqn:SimpleEV}
\end{equation}
At the supersymmetric point $x=1$, it has multiplicity $1$ and is the Hamiltonian's ground-state eigenvalue \cite{hagendorf:17}. By continuity, this property remains valid for real $x$ sufficiently close to $x=1$, which motivates our focus on the corresponding eigenspace.
In \cite{hagendorf:21,hagendorf:22}, we have constructed a special eigenvector whose components (with respect to a standard basis of $(\mathbb C^2)^{\otimes N}$ which we recall below) are polynomials in $x$ with integer coefficients. Some of these polynomials appear as generating functions for the enumeration of ASMs or certain rhombus tilings. The main technical tool for this work is a solution to the boundary quantum Knizhnik-Zamolodchikov (bqKZ) equations associated with the $R$- and a $K$-matrix of a six-vertex model \cite{cherednik:92,jimbo:95,difrancesco:07,stokman:15,reshetikhin:18}. This solution is built from multivariate Laurent polynomials that are explicitly given in terms of multiple contour integrals. A specialisation of the solution's parameters, called the homogeneous limit, allows one to obtain and characterise the special eigenvector of the spin-chain Hamiltonian.

The main goal of this article is to analyse the sum of the special eigenvector's components, that we denote by $S_N$. This sum is a polynomial in $x$ with integer coefficients. At the supersymmetric point $x=1$ and for $N=1,\dots,5$, it evaluates to
\begin{equation}
\label{eqn:SumCompsSUSY}
      S_1 = 1,\quad S_2 = 2,\quad S_3= 4,\quad S_4= 13,\quad S_5 = 46.
\end{equation}
These integers correspond to the first entries of the sequence of the numbers of totally-symmetric alternating sign matrices (TSASMs) of odd order \cite{oeistsasm:24}. TSASMs are ASMs that are invariant under all symmetries of the square. Despite their maximal symmetry, results on the TSASM enumeration are scarce in the literature \cite{bousquet:95,robbins:00,kuperberg:02,behrend:23}.  Their appearance in relation with the ground-state eigenvector of the open XXZ spin chain at the supersymmetric point comes, to a certain extent, as a surprise. 
In \cref{cor:SumCompsSpecialEV}, we provide an expression for $S_N$ with arbitrary $N$ and $x$ in terms of a generating function for a weighted enumeration of TSASMs of order $2N+1$.
Exploiting the properties of this TSASM generating function, we find $S_N$ at the supersymmetric point for arbitrary $N$ in \cref{corr:SumCompsSUSY}, thereby generalising \eqref{eqn:SumCompsSUSY}. Moreover, a remarkable consequence of our main result is a multiple contour-integral formula for the number of TSASMs of order $2N+1$ stated in \cref{corr:TSASMContourIntegral}. To the best of our knowledge, it provides the only known formula to date for the number of TSASMs of arbitrary order.

To prove our results, we consider multivariate generalisations of both the sum of components of the special eigenvector and the TSASM generating function. The first quantity is the \textit{generalised sum of components}, which we define and analyse in \cref{sec:GeneralisedSumComps}. It is a linear combination of the components of the Laurent-polynomial solution to the bqKZ equations mentioned above. In particular, we show in \cref{prop:Uniqueness} that it is uniquely characterised by a set of analyticity and symmetry properties, as well as several reduction relations. The second quantity is a linear combination of partition functions of an inhomogeneous six-vertex model on a square-grid graph with a triangular shape, which we study in \cref{sec:6VTSASM}. For certain boundary conditions, the configurations of this six-vertex model are in bijection with TSASMs. We show that the linear combination of partition functions possesses the same properties as the generalised sum of components (up to suitable parameter identifications and modulo trivial factors). Using the uniqueness property, we establish a relation between both quantities in \cref{thm:YandYY}. In \cref{sec:MainResult}, we use this relation to prove \cref{thm:MainTheorem}, which states the equality between a two-parameter generalisation of the sum of components and the TSASM generating function. \Cref{cor:SumCompsSpecialEV,corr:TSASMContourIntegral,corr:SumCompsSUSY} immediately follow from this theorem.

\section{The sum of components and its generalisation}
\label{sec:GeneralisedSumComps}

The purpose of this section is to develop the technical tools that allow us to compute the sum of the components of the special eigenvector. In \cref{sec:MCI}, we recall the construction of this eigenvector in terms of certain multiple contour integrals. \Cref{sec:PropMCI} provides a collection of useful properties of these contour integrals. In \cref{sec:DefGenSum}, we define the generalised sum of components, which is a multivariate generalisation of the sum of components of the special eigenvector. We investigate its properties in \cref{sec:PropGenSum} and show that they determine it uniquely.

\subsection{From multiple contour integrals to the special eigenvector}
\label{sec:MCI}

Throughout, $N\geqslant 0$ denotes an integer and we systematically use the integers $n,n',\epsilon$, defined through
\begin{equation}
  \label{eqn:Defnnbar}
  n = \left\lfloor N/2\right\rfloor, \quad n' = \left\lceil N/2\right\rceil, \quad \epsilon = n' - n.
\end{equation}
We denote by $\mathbb C^\times$ the set $\mathbb C\backslash \{0\}$.
 Moreover, in this and the following sections, we adopt the abbreviations 
\begin{equation}
\label{eqn:NotationBrackets}
  \bar z = z^{-1}, \quad [z]=z-{\bar z},\quad \{z\} = z+{\bar z},
\end{equation}
for each $z\in \mathbb C^\times$.

\subsubsection{Multiple contour integrals}
For each $N\geqslant 2$ and each sequence $a_1,\dots,a_n$ of integers satisfying $1\leqslant a_1 <\cdots < a_n \leqslant N$, we define the multiple contour integral
\begin{multline}
  (\Psi_N)_{a_1,\dots,a_n} = (-[q])^n \prod_{1\leqslant i < j \leqslant N} [q z_j{\bar z}_i][q^2 z_i z_j]\\
  \times \oint\cdots \oint\Xi_{a_1,\dots,a_n}(w_1,\dots, w_n|z_1,\dots,z_N) \prod_{i=1}^n \frac{\diff w_i}{\pi \ii w_i},
  \label{eqn:CIBigPsi}
\end{multline}
where 
\begin{multline}
  \Xi_{a_1,\dots,a_n}(w_1,\dots, w_n|z_1,\dots,z_N)\\
  =\frac{\prod_{1\leqslant i < j \leqslant n}[q w_j{\bar w}_i][w_i{\bar w}_j][q w_iw_j][q^2 w_i w_j]\prod_{i=1}^n[q^2 w_i^2][\beta w_i]}{\prod_{i=1}^n\left(\prod_{j=1}^{a_i}[z_j{\bar w}_i]\prod_{j=a_i}^N[qz_j{\bar w}_i]\prod_{j=1}^N[q^2 w_i z_j]\right)}.
\end{multline}
Here and throughout,
we assume $z_1,\dots,z_N,\beta,q\in \mathbb C^\times$ with the restrictions $q^4 \neq 1$ and $\beta^2 \neq 1$. The integration contour of each $w_i$ in \eqref{eqn:CIBigPsi} surrounds the simple poles $z_1,\dots,z_N$, but not other poles of the integrand. 

The evaluation of the multiple contour integral \eqref{eqn:CIBigPsi} is, in principle, possible with the help of the residue theorem. For small values of $N$, this evaluation is straightforward.
\begin{example}
\label{ex:Psi23}
For $N=2$, we obtain
\begin{equation}
  (\Psi_2)_1 = [\beta z_1], \quad (\Psi_2)_2 = -[q\beta z_2].
\end{equation}
For $N=3$, we obtain the two factorised expressions 
\begin{equation}
  \left(\Psi_3\right)_1= [\beta z_1][q z_3{\bar z}_2][q^2 z_2 z_3], \quad \left(\Psi_3\right)_3 = [q\beta z_3][q z_2{\bar z}_1][q z_1 z_2],
\end{equation}
and 
\begin{equation}
   \left(\Psi_3\right)_2 = \left([q][\beta z_1][q z_3{\bar z}_2][q^2 z_2 z_3]-[\beta z_2][q z_2{\bar z}_1][qz_3{\bar z}_1][q^2 z_1z_3]\right)/[z_2{\bar z}_1].
\end{equation}
\end{example}

We sometimes write $(\Psi_N)_{a_1,\dots,a_n}(z_1,\dots,z_N)$ to stress the dependence on $z_1,\dots,z_N$. The point $z_1=\dots=z_N=1$ is commonly referred to as the homogeneous limit. To discuss this point, we introduce the rescaled quantities
\begin{equation}
  \label{eqn:LittleBigPsi}
  (\psi_N)_{a_1,\dots,a_n} = (-1)^{n'(n'-1)/2}[\beta]^{-n}[q]^{-n(n-1)-n'(n'-1)}(\Psi_N)_{a_1,\dots,a_n}(1,\dots,1).
\end{equation}
They are also given by multiple contour integrals \cite[Section 5.2]{hagendorf:21}, 
\begin{multline}
  \label{eqn:CILittlePsi}
  (\psi_N)_{a_1,\dots,a_n} = \oint \frac{\diff u_1}{2\pi \ii }\cdots \oint \frac{\diff u_n}{2\pi \ii} \prod_{k=1}^n \frac{(u_k+x)(1+\tau u_k + u_k^2)^{\epsilon}}{u_k^{N+1-a_{n+1-k}}}\\
  \times \prod_{1\leqslant i\leqslant j \leqslant n}(1-u_iu_j)\prod_{1\leqslant i < j \leqslant n} (u_j-u_i)(1+\tau u_j+u_iu_j)(\tau + u_i+u_j),
\end{multline}
where the integration contour of each integral is a simple, positively-oriented closed curve around $0$. Moreover, we have
\begin{equation}
\label{eqn:ParametersXTau}
  x= -[\beta q]/[\beta], \quad \tau = -\{q\}.
\end{equation}
The evaluation of \eqref{eqn:CILittlePsi} is, in principle, straightforward and yields a polynomial in $x$ and $\tau$ with integer coefficients. 
\begin{example} 
For $N=4$, we obtain
\begin{align}
  (\psi_4)_{1,2} &= \tau, &(\psi_4)_{2,3} &= x + \tau + x^2 \tau,\\
  (\psi_4)_{1,3} &= 1+x\tau+\tau^2, &(\psi_4)_{2,4} &= x(x + \tau + x\tau^2),\\  
  (\psi_4)_{1,4} &= x(1+\tau^2), &(\psi_4)_{3,4} &= x^2\tau. 
\end{align}
\end{example}

For arbitrary $N$ and generic $a_1,\dots,a_n$, we do not possess an explicit formula for the polynomials defined through \eqref{eqn:CILittlePsi}. However, for $a_i=i$, we have \cite[Section 5.2]{hagendorf:21}
\begin{equation}
  \label{eqn:SimpleComponents}
  (\psi_N)_{1,2,\dots,n} = \tau^{n'(n'-1)/2}
\end{equation}
for all $N\geqslant 0$.

\subsubsection{The special eigenvector} We now recall how to construct the special eigenvector of the Hamiltonian \eqref{eqn:XXZHamiltonian} with the help of the multiple contour integrals. Let
\begin{equation}
  \ket{\uparrow}= \begin{pmatrix}
  1\\ 0
  \end{pmatrix},
  \quad
  \ket{\downarrow} = \begin{pmatrix}
  0\\ 1
  \end{pmatrix},
\end{equation}
be the standard basis vectors of $\mathbb C^2$. For $N \geqslant 1$ and each spin sequence $\bm \alpha = \alpha_1\cdots\alpha_N$ with $\alpha_i \in \{\uparrow,\downarrow\}$, $i=1,\dots,N$, we define
\begin{equation}
  \ket{\bm \alpha} = \ket{\alpha_1} \otimes \cdots \otimes \ket{\alpha_N}.
  \label{eqn:StandardBasis}
\end{equation}
The set of these vectors defines the standard basis of $(\mathbb C^2)^{\otimes N}$. We refer to the coefficients of any vector in this basis as its \emph{components}. 

Let $|\psi_1\rangle = \ket{\uparrow}$ and define, for each $N\geqslant 2$, the vector 
\begin{equation}
  |\psi_N\rangle = \sum_{1\leqslant a_1 < \dots < a_n \leqslant N}(\psi_N)_{a_1,\dots,a_n}\ket{\uparrow\cdots \uparrow\underset{a_1}{\downarrow} \uparrow \quad \underset{\cdots}{\cdots} \quad \uparrow \underset{a_n}{\downarrow} \uparrow \cdots \uparrow}\in (\mathbb C^2)^{\otimes N},
  \label{eqn:DefLittlePsi}
\end{equation}
where we used underset indices to indicate the position of the $\downarrow$'s. The multiple contour integrals \eqref{eqn:CILittlePsi} thus define the components of this vector (that do not trivially vanish). Furthermore, we note that, by \eqref{eqn:SimpleComponents}, $|\psi_N\rangle$ does not vanish identically.

By \cite[Theorem 2]{hagendorf:21}, if $\tau =1$, we have
\begin{equation}
   H|\psi_N\rangle = E|\psi_N\rangle, \quad M|\psi_N\rangle = \frac{\epsilon}{2}|\psi_N\rangle,
\end{equation}
for each $N\geqslant 1$, where $H$ is the spin-chain Hamiltonian \eqref{eqn:XXZHamiltonian}, $E$ is the eigenvalue \eqref{eqn:SimpleEV}, and
\begin{equation}
  M = \frac12 \sum_{i=1}^N\sigma_i^z
\end{equation}
is the magnetisation operator, which commutes with the Hamiltonian. By \eqref{eqn:SimpleComponents}, the eigenvector's normalisation is fixed by
\begin{equation}
  (\psi_N)_{1,2,\dots,n} = 1.
\end{equation}

Moreover, if $x$ is a real positive number then $E$ is the lowest eigenvalue of $H$, when restricted to the eigenspace of $M$ with eigenvalue $\epsilon/2$, and it has multiplicity $1$ \cite[Theorem 1]{hagendorf:21}. As mentioned in the introduction, for real $x$ sufficiently close to the supersymmetric point $x=1$, it is the lowest eigenvalue of $H$, even without the restriction to this eigenspace. Numerical investigations suggest that this statement holds even for all $x>0$, but a proof is currently missing.

\subsection{Properties of the contour integrals}
\label{sec:PropMCI}
Explicitly evaluating the multiple contour integrals \eqref{eqn:CIBigPsi} becomes increasingly complicated as $N$ grows. Instead, we analyse their properties as functions of $z_1,\dots,z_N$. We now provide these properties in a series of lemmas that summarise the results of \cite{hagendorf:21,hagendorf:22}.

\subsubsection{Laurent polynomials.}
We recall that a Laurent polynomial in a single indeterminate $z$
is an expression of the form
\begin{equation}
  f(z) = \sum_{k=-\infty}^{\infty}a_k z^k,
\end{equation}
where only finitely many coefficients are non-zero. For non-zero $f$, we may define the degrees of its leading and trailing terms as $d= \max\{k|a_k\neq 0\}$ and $d' = \min\{k|a_k \neq 0\}$, respectively. The degree width of $f$ is the difference $d-d'$. Moreover, we say that $f$ is centred if $d+d'=0$, in which case its degree width is even. Conversely, if $f$ vanishes identically, then $d$ and $d'$ are undefined. In this case, we declare that the Laurent polynomial is centred, but define its degree width as $-\infty$. Finally, a Laurent polynomial in several indeterminates $z_1,\dots,z_N$ is a Laurent polynomial in each $z_i$.

The following lemma characterises $(\Psi_N)_{a_1,\dots,a_n}$ as a Laurent polynomial in $z_1,\dots,z_N$ (see \cite[Propositions 3.11, 3.12 and 3.14]{hagendorf:21}):
\begin{lemma}
  \label{lem:LaurentPsi} 
  For $N\geqslant 2$ and each $i=1,\dots, N$, $(\Psi_N)_{a_1,\dots,a_n}$ is a centred Laurent polynomial in $z_i$ of degree width at most $4(n'-1)$ if $i \notin \{a_1,\dots,a_n\}$ and at most $2(2n-1)$ if $i\in  \{a_1,\dots,a_n\}$.
\end{lemma}

\subsubsection{The exchange and reflection relations}
Let $\ket{\Psi_1} = \ket{\uparrow}$ and define, for each $N\geqslant 2$,
\begin{equation}
  |\Psi_N\rangle = \sum_{1\leqslant a_1 < \dots < a_n \leqslant N}(\Psi_N)_{a_1,\dots,a_n}\ket{\uparrow\cdots \uparrow\underset{a_1}{\downarrow} \uparrow \quad \underset{\cdots}{\cdots} \quad \uparrow \underset{a_n}{\downarrow} \uparrow \cdots \uparrow}.
\end{equation}
If needed, we write $|\Psi_N\rangle = |\Psi_N(z_1,\dots,z_N)\rangle$ to stress the dependence on $z_1,\dots,z_N$.

Let $z\in \mathbb C^\times$ be an indeterminate. We define a $\check R$-matrix of the so-called six-vertex model as the linear operator $\check R(z) \in \mathrm{End}(\mathbb C^2 \otimes \mathbb C^2)$ with the following matrix representation with respect to the standard basis $\{\ket{\uparrow\uparrow},\ket{\uparrow\downarrow},\ket{\downarrow\uparrow},\ket{\downarrow\downarrow}\}$:\footnote{We do not distinguish between linear operators and their matrix representations with respect to the standard basis.}
\begin{equation}
  \check R(z) =
  \frac{1}{[q{\bar z}]}
  \begin{pmatrix}
    [q z] & 0 & 0 & 0\\
    0 & [q] & [z] & 0\\
    0 & [z] & [q] & 0\\
    0 & 0 & 0 & [qz]
  \end{pmatrix}.
\end{equation}
It obeys the so-called (braid) Yang-Baxter equation
\begin{equation}
  \label{eqn:BraidYBE}
  \check R_{2,3}(z_1{\bar z}_2)\check R_{1,2}(z_1{\bar z}_3)\check R_{2,3}(z_2{\bar z}_3)= 
	\check R_{1,2}(z_2{\bar z}_3)  \check R_{2,3}(z_1{\bar z}_3)\check R_{1,2}(z_1{\bar z}_2),
\end{equation}
on $\mathbb C^2 \otimes \mathbb C^2 \otimes \mathbb C^2$ for all $z_1,z_2,z_3\in \mathbb C^\times$ such that both sides of the equality are well-defined. Here (and throughout), we use the standard tensor-leg notation to indicate the tensor factors on which the $\check R$-matrix or other operators act non-trivially. The following lemma restates the result of \cite[Proposition 3.2]{hagendorf:21}:
\begin{lemma}[Exchange relations]
\label{lem:Exchange}
For $N\geqslant 2$ and each $i=1,\dots, N-1$, we have
\begin{equation}
  \label{eqn:Exchange}
  \check R_{i,i+1}(z_i\bar{z}_{i+1})|\Psi_N(\dots,z_i,z_{i+1},\dots)\rangle = |\Psi_N(\dots,z_{i+1},z_i,\dots)\rangle.
\end{equation}
\end{lemma}

We also need a $K$-matrix associated with the $\check R$-matrix of the six-vertex model. It is a linear operator $K(z)\in \mathrm{End}(\mathbb C^2)$, whose matrix representation with respect to the standard basis is 
\begin{equation}
  K(z) =
  \begin{pmatrix}
    1 & 0\\
    0 & [\beta z]/[\beta {\bar z}]
  \end{pmatrix},
\end{equation}
where $z\in \mathbb C^\times$ is an indeterminate. The $K$-matrix obeys the boundary Yang-Baxter equation \cite{sklyanin:88}
\begin{align}
  \check R_{1,2}(z_1{\bar z}_2)K_1(z_1)\check R_{1,2}(z_1z_2)K_1(z_2)&=K_1(z_2) \check R_{1,2}(z_1z_2) K_1(z_1)\check R_{1,2}(z_1{\bar z}_2),
  \label{eqn:Reflection}
\end{align}
which holds on $\mathbb C^2 \otimes \mathbb C^2$ for all $z_1,z_2\in \mathbb C^\times$ such that both sides of the equality are well-defined. The next lemma recalls a result of \cite[Proposition 3.7]{hagendorf:21}:
\begin{lemma}[Left reflection relation]
\label{lem:Reflection}
\begin{subequations}
  We have
\begin{align}
  K_1(\bar{z}_1)|\Psi_N(z_1,z_2,\dots,z_N)\rangle  &= |\Psi_N(\bar{z}_1,z_2,\dots,z_N)\rangle.\end{align}
\end{subequations}
\end{lemma}
The compatibility of the exchange relations and the left reflection relation follows from \eqref{eqn:BraidYBE} and \eqref{eqn:Reflection}. Moreover, the vector $|\Psi_N\rangle$ obeys a right reflection relation, too. Together, these relations imply that $|\Psi_N\rangle$ is a solution to the bqKZ equations, as stated in the introduction. We refer to \cite{hagendorf:21} for the details.

\subsubsection{Reduction}

Let $\ket{\zeta} =|{\uparrow\downarrow}\rangle - |{\downarrow\uparrow}\rangle$ and define, for each $N\geqslant 3$ and $i=1,\dots,N-1$, the linear operator $\Xi_N^i:(\mathbb C^2)^{\otimes(N-2)}\to (\mathbb C^2)^{\otimes N}$ through the following action on the standard basis vectors of $(\mathbb C^2)^{\otimes(N-2)}$:
\begin{equation}
  \Xi_N^i|\alpha_1\cdots \alpha_{N-2}\rangle = |\alpha_1\cdots \alpha_{i-1}\rangle \otimes \ket{\zeta} \otimes |\alpha_i\cdots \alpha_{N-2}\rangle.
\end{equation}
Here, $\alpha_1,\dots,\alpha_{N-2}\in \{\uparrow,\downarrow\}$. The following result was obtained in \cite[Proposition 2.2]{hagendorf:22}:
\begin{lemma}[Reduction relation]
\label{lem:PsiReduction}
We have $|\Psi_2(z_1,z_2={\bar q}z_1)\rangle = -[\beta z_1]\ket{\zeta}$
and, for each $N\geqslant 3$ and $i =1,\dots, N-1$, 
\begin{multline}
  |\Psi_N(\dots,z_i,z_{i+1}={\bar q}z_i,\dots)\rangle=(-1)^{n+i+1}[\beta z_i]\prod_{j=1}^{i-1}[qz_i{\bar z}_j][q z_i z_j]\\
    \times \prod_{j=i+2}^N
[ q^2{\bar z}_i{z}_j]
[qz_iz_j]\Xi^i_N|\Psi_{N-2}(\dots,\widehat{z_i,z_{i+1}},\dots)\rangle, 
\end{multline}
where $\widehat{\cdots}$ denotes omission.
\end{lemma}

\subsection{The generalised sum of components}
\label{sec:DefGenSum}
For each $N\geqslant 2$, we define
\begin{equation}
  \label{eqn:DefS}
  S_N = \sum_{1\leqslant a_1 < \cdots < a_n \leqslant N} (\psi_N)_{a_1,\dots,a_n},
\end{equation}
where $(\psi_N)_{a_1,\dots,a_n}$ are the polynomials defined in \eqref{eqn:CILittlePsi}. For convenience, we also define $S_0=S_1=1$. Clearly, $S_N$ represents the sum of the components of the vector $|\psi_N\rangle$ for each $N\geqslant 1$. For $\tau =1$, it is the sum of components of the Hamiltonian's special eigenvector, whose computation is the goal of this work.

Just as the components of $|\psi_N\rangle$, the sum of components $S_N$ possesses a multiple contour-integral formula \cite{hagendorf:21}:
\begin{multline}
  \label{eqn:CIS}
   S_N=\oint\cdots \oint \prod_{k=1}^n\frac{\diff u_k}{2\pi \ii}\frac{(u_k+x)(1+\tau u_k + u_k^2)^{\epsilon}}{u_k^{n'+k}\left(1-\prod_{j=1}^ku_j\right)}\\
  \times \prod_{1\leqslant i \leqslant j \leqslant n}(1-u_iu_j)\prod_{1\leqslant i < j \leqslant n}(u_j-u_i)(\tau + u_i + u_j)(1+\tau u_j + u_i u_j).
  \end{multline}
Here, the integration contour of each integral is a simple, positively-oriented curve around $0$, but no other singularities of the integrand. The evaluation of this multiple contour integral is an application of the residue theorem, which yields a polynomial with integer coefficients in $x$ and $\tau$.
\begin{example}
 Using \eqref{eqn:CIS}, we obtain
  \begin{align}
    S_2 &= 1+x, \\
    S_3 &= (1+x)(1+\tau),\\
    S_4 &= x+(1+x+x^2)(1+\tau)^2,\\
    S_5 &= x(1+\tau)^2(2+2\tau+\tau^2)+(1+x^2)(1+4\tau+4\tau^2+3\tau^3+\tau^4).
  \end{align}
  For $x=0$ and $\tau=1$, we find
  \begin{equation}
   S_0 = S_1 = S_2 = 1,\quad S_3 = 2,\quad S_4= 4,\quad S_5= 13.
  \end{equation}
  Similarly, if $x=1$ and $\tau=1$, the polynomials evaluate to
  \begin{equation}
     S_0 = S_1 = 1,\quad S_2 = 2,\quad S_3= 4,\quad S_4= 13,\quad S_5 = 46,
  \end{equation}
  as was anticipated in \eqref{eqn:SumCompsSUSY}. In both cases, we recognise the first entries of the sequence of numbers enumerating TSASMs \cite{oeistsasm:24}.
\end{example}

To investigate $S_N$, we introduce the \textit{generalised sum of components}. To this end, we define, for each vector $\ket{\psi} \in (\mathbb C^2)^{\otimes N}$, a corresponding covector $\bra {\psi} \in ((\mathbb C^2)^{\otimes N})^\ast$ via transposition, $\bra{\psi} = |\psi\rangle^{\mathsf T}$, without complex conjugation. Let
\begin{equation}
  \bra{\chi(w)} = \frac{\{q^{1/2}w\}}{\{q^{1/2}\}}(\bra{ \uparrow\uparrow} + \bra{\downarrow\downarrow}) 
+ \bra{\uparrow\downarrow}+ \bra{\downarrow\uparrow}, \quad \bra{\varphi} = \bra{\uparrow}+\bra{\downarrow},
\end{equation}
where $w \in \mathbb C^\times$.\footnote{Here, $q^{1/2}\in \mathbb C^\times$ is one of the two solutions of $(q^{1/2})^2=q$. The specific choice does not influence our results. In the following, we write $q^{k/2}$ instead of $(q^{1/2})^k$ for each integer $k$.} Moreover, for $N\geqslant 2$ and $w_1,\dots,w_n\in \mathbb C^\times$, we use the abbreviation
\begin{equation}
 \bra{\chi_N(w_1,\dots,w_n)} = \bigotimes_{i=1}^n \bra{\chi(w_i)} \in  ((\mathbb C^2)^{\otimes 2n})^\ast.
\end{equation}
We call \emph{overlap} the dual pairing $\bra{\phi}\psi\rangle$ of a covector $\bra{\phi}$ and a vector $\ket{\psi}$. For each $N\geqslant 2$, the generalised sum of components $Z_N$ is a function of $w_1,\dots,w_n$ (and, implicitly, of $\beta$ and $q$), given by the overlap
\begin{subequations}
\label{eqn:DefZ}
\begin{equation}
  Z_{N}(w_1,\dots,w_n) = \langle\chi_N(w_1,\dots,w_n)|\Psi_{N}(w_1,{\bar w}_1,\dots, w_n,{\bar w}_n)\rangle,
\end{equation}
if $N$ is even, and by
\begin{equation}
   Z_{N}(w_1,\dots,w_n)
   =\left( \langle\chi_N(w_1,\dots,w_n)|\otimes \langle \varphi|\right) |\Psi_{N}(w_1,{\bar w}_1,\dots, w_n,{\bar w}_n,1)\rangle.
\end{equation}
\end{subequations}%
if $N$ is odd. Moreover, it will be convenient to set $Z_0=Z_1=1$. The choice $z_{2i-1}=w_i,\,z_i = \bar w_i$ for $i=1,\dots,n$, and $z_N=1$ for odd $N$, for the arguments of $\ket{\Psi_N}$ is sometimes referred to as \textit{half-specialisation} and often leads to important simplifications \cite{kuperberg:02,zinnjustin:13,hagendorf:16,morin:20,brasseur:21}.

\begin{example}
Using the expressions from \cref{ex:Psi23}, we obtain
\begin{equation}
  \label{eqn:Z2}
  Z_2(w_1) = [q^{-1/2}w_1]\{q^{1/2}\beta\}
\end{equation}
and
\begin{equation}
\label{eqn:Z3}
Z_3(w_1) =\frac{[q^{-1/2}w_1][qw_1][q{\bar w}_1]\{q^{1/2}\beta\}\{q^{3/2}\}}{\{q^{1/2}\}}.
\end{equation} 
\end{example}

In the next proposition, we show how to obtain $S_N$ from $Z_N$ in the homogeneous limit $w_1=\dots=w_n=1$:
\begin{proposition}\label{prop:relationSZ}
 For $N\geqslant 0$ and $x,\tau$ given by \eqref{eqn:ParametersXTau}, we have
\begin{equation}\label{eqn:relationSZ}
S_N= (-1)^{n'(n'-1)/2}[\beta]^{-n} [q]^{-n(n-1)-n'(n'-1)} Z_{N}(1,\dots,1).
\end{equation}
\end{proposition}
\begin{proof}
The proof is trivial for $N=0,1$. For $N\geqslant 2$, we use the factorisation
\begin{equation}
 \bra{\chi(1)} = \bra{\uu\uu}+\bra{\uu\dd}+\bra{\dd\uu}+\bra{\dd\dd} = \bra{\varphi} \otimes \bra{\varphi},
\end{equation}
to obtain
\begin{equation}
  Z_N(1,\dots,1) = \left(\langle \varphi|\otimes \cdots \otimes \langle \varphi|\right) \ket {\Psi_N(1,\dots,1)}.
\end{equation}
Moreover, \eqref{eqn:LittleBigPsi} leads to
\begin{equation}
\ket{\psi_N} = (-1)^{n'(n'-1)/2}[\beta]^{-n} [q]^{-n(n-1)-n'(n'-1)} \ket{\Psi_N(1,\dots,1)}  ,
\end{equation}
where the parameters $x$ and $\tau$ on the left-hand side are given by \eqref{eqn:ParametersXTau}. Hence,
\begin{equation}
(-1)^{n'(n'-1)/2}[\beta]^{-n} [q]^{-n(n-1)-n'(n'-1)}Z_N(1,\dots,1)= \left(\bra{\varphi}\otimes \cdots \otimes \bra{\varphi}\right) \ket {\psi_N}.
\end{equation}
Using the explicit form of $\bra{\varphi}$, we conclude that $\left(\bra{\varphi}\otimes \cdots \otimes \bra{\varphi}\right) \ket {\psi_N}=S_N$.
\end{proof}

\subsection{Properties}
\label{sec:PropGenSum}
For $N\geqslant 4$, the explicit expression of $Z_N$ becomes increasingly complex. Fortunately, the properties of both 
$\bra{\chi(w)}$ and $\ket{\Psi_N(z_1,\dots,z_N)}$ enable us to characterise the generalised sum of components in terms of its analyticity and symmetry properties in $w_1,\dots,w_n$, and through a set of reduction relations. The purpose of this section is to obtain these properties and show that they determine $Z_N$ uniquely.

\subsubsection{Symmetries and degree}
We will make frequent use of the following property:
\begin{proposition}
  \label{prop:ZSymmetry}
  $Z_N$ is a symmetric function of $w_1,\dots,w_n$.
\end{proposition}
\begin{proof}
  The proposition follows from a standard calculation \cite{morin:20}, based on the identity  \begin{equation}
    \label{eqn:bYBEChi}
    \langle \chi(w)|\otimes \langle \chi(z)| \check R_{2,3}(zw) \check R_{1,2}(z {\bar w})=\langle \chi(z)|\otimes \langle \chi(w)| \check R_{2,3}(zw) \check R_{3,4}(z {\bar w}),
  \end{equation}
  which is readily verified, and the exchange relations \eqref{eqn:Exchange}.
\end{proof}

In the following, we also need a bound on the degree width of the Laurent polynomial $(\Psi_N)_{a_1,a_2,\dots}(w_1,{\bar w}_1,\dots)$ in $w_1$. This bound depends on the integers $a_1,\dots,a_n$.
\begin{lemma}
  \label{lem:DegreeHSPsi}
  For $N\geqslant 4$, $(\Psi_N)_{a_1,a_2\cdots}(w_1,{\bar w}_1,\dots)$ is a  centred Laurent polynomial in $w_1$ of degree width at most $4(N-2)$ if \textit{(i)} $a_1=1, a_2=2$ or \textit{(ii)} $a_1\geqslant 3$, and at most $2(2N-3)$ if \textit{(iii)} $a_1 = 1, a_2\geqslant 3$ or \textit{(iv)} $a_1=2$.
  \end{lemma}
\begin{proof}
 Note that $n,n'\geqslant 2$ for $N\geqslant 4$. We consider each of the four cases \textit{(i)}-\textit{(iv)} separately,  using \Cref{lem:LaurentPsi,lem:Exchange,lem:Reflection} and the relations 
 \begin{equation}
  \label{eqn:nnbarInequalities}
  n+n' = N,\quad n \leqslant n', \quad n' \leqslant n+1,
\end{equation}
which are an easy consequence of \eqref{eqn:Defnnbar}.
 
 \noindent \textit{Case (i):} The lemmas imply that
\begin{equation}
  (\Psi_N)_{1,2,\dots}(z_1,z_2,\dots) = [\beta z_1][\beta z_2][qz_2{\bar z}_1][qz_1z_2](\Phi_N)_{1,2,\dots}(z_1,z_2,\dots),
\end{equation}
where $(\Phi_N)_{1,2,\dots}(z_1,z_2,\dots)$ is a centred Laurent polynomial in both $z_1$ and $z_2$ of degree width $4(n-2)$. Specialising $z_1=w_1,z_2={\bar w}_1$, we find
\begin{equation}
(\Psi_N)_{1,2,\dots}(w_1,{\bar w}_1,\dots) = [\beta w_1][\beta {\bar w}_1][q{\bar w}^2_1][q](\Phi_N)_{1,2,\dots}(w_1,{\bar w}_1,\dots).
\end{equation}
The right-hand side is a Laurent polynomial in $w$ of degree width at most $8+2\times 4(n-2)=8(n-1)$. By \eqref{eqn:nnbarInequalities}, $8(n-1)$ is bounded by $4(N-2)$.
  
 \medskip
 \noindent \textit{Case (ii):} In this case, we use the lemmas to write
  \begin{equation}
    (\Psi_N)_{a_1,a_2,\dots}(z_1,z_2,\dots) = [qz_2{\bar z}_1][qz_1z_2](\tilde \Phi_N)_{a_1,a_2\dots}(z_1,z_2,\dots)
  \end{equation}
  where $(\tilde\Phi_N)_{a_1,a_2,\dots}(z_1,z_2,\dots)$ is a centred Laurent polynomial in both $z_1,z_2$ of degree width at most $4(n'-2)$. The specialisation $z_1=w_1,z_2={\bar w}_1$ leads to
  \begin{equation}
    (\Psi_N)_{a_1,a_2,\dots}(w_1,{\bar w}_1,\dots) = [q{\bar w}^2_1][q](\tilde\Phi_N)_{a_1,a_2,\dots}(w_1,{\bar w}_1,\dots).
  \end{equation}
  The right-hand side is a centred Laurent polynomial in $w$ of degree width at most $4(2n' -3)$, which is bounded by $4(N-2)$, too, thanks to \eqref{eqn:nnbarInequalities}.

\medskip
\noindent
\textit{Cases (iii) and (iv):} In these cases, \cref{lem:LaurentPsi} allows us to conclude directly that the coefficient $(\Psi_N)_{a_1,a_2,\dots}(w_1,{\bar w}_1,\dots)$ is a centred Laurent polynomial in $w_1$ of degree width at most $4(n'-1)+2(2n-1)=2(2(n+n')-3)=2(2N-3)$.
\end{proof}

\begin{proposition}
  \label{prop:ZDegree}
  For $N\geqslant 2$ and $i=1,\dots,n$, $Z_N$ is a centred Laurent polynomial in $w_i$ of degree width at most $2(2N-3)$.
\end{proposition}
\begin{proof}
  For $N=2,3$, the proposition follows from \eqref{eqn:Z2} and \eqref{eqn:Z3}. For $N\geqslant 4$, we infer from \cref{prop:ZSymmetry} that it suffices to consider $i=1$.  By \eqref{eqn:DefZ}, $Z_N$ is a linear combination of the components $(\Psi_N)_{a_1a_2\dots}(w_1,{\bar w}_1,\dots)$. We examine it in each of the four cases \textit{(i)}-\textit{(iv)} of \cref{lem:DegreeHSPsi}. For \textit{(i)} $a_1=1, a_2=2$ and \textit{(ii)} $a_1\geqslant 3$, the component's coefficient in the linear combination is proportional to  $\{q^{1/2} w_1\}$, which is a centred Laurent polynomial of degree width $2$ in $w_1$. In the cases \textit{(iii)} $a_1 = 1, a_2\geqslant 3$ or \textit{(iv)} $a_1=2$, the component's coefficient in the linear combination is independent of $w_1$. For each of the four cases, \cref{lem:DegreeHSPsi} allows us to conclude that the product of the coefficient and the component is a centred Laurent polynomial in $w_1$ of degree width at most $2(2N-3)$. The same holds for the linear combination.
\end{proof}

\begin{proposition}
 \label{prop:ZSymmetry2}
 For $N \geqslant 2$ and $i=1,\dots,n$, we have
  \begin{align}
    \label{eqn:ZReflection}
    Z_N(\dots,-w_i,\dots) &= - Z_N(\dots,w_i,\dots),\\
    \label{eqn:ZInversion}
    [q^{-1/2}w_i]Z_N(\dots,{\bar w}_i,\dots) &=[q^{-1/2}{\bar w}_i]Z_N(\dots,w_i,\dots).
  \end{align}
\end{proposition}
\begin{proof}
   The relation \eqref{eqn:ZReflection} follows from the identity  $\bra{\chi(-z)}= - \bra{\chi(z)} (\sigma^z\otimes \sigma^z)$, and from
  \begin{equation}
  \label{eqn:PsiReflection}
    |\Psi_N(\dots,-z_i,\dots)\rangle=\sigma_i^z|\Psi_N(\dots,z_i,\dots)\rangle,
  \end{equation}
  which is an immediate consequence of \cite[Propositions 3.11 and 3.12]{hagendorf:21}. Moreover, \eqref{eqn:ZInversion} follows from \cref{lem:Exchange} and the relation
  \begin{equation}
    \langle \chi({\bar z})|\check R(z^2) = \frac{[q^{-1/2}{\bar z}]}{[q^{-1/2}z]}\langle \chi(z)|,
  \end{equation}
  which results from a simple calculation. 
\end{proof}
This proposition allows us to identify two zeros of the generalised sum of components:
\begin{corollary}
  \label{prop:ZSpecialZeros}
  For $N\geqslant 2$ and each $i=1,\dots,n$, we have $Z_{N}(\dots,w_i = \pm q^{1/2},\dots) = 0$.
\end{corollary}
For odd $N$, we identify two further zeros:
 \begin{proposition}
  \label{prop:ZOddExtraZero}
  For odd $N\geqslant 3$ and each $i=1,\dots,n$, we have $Z_{N}(\dots,w_i = \pm {\bar q},\dots) = 0$.
\end{proposition}
\begin{proof}
  By \cref{prop:ZSymmetry}, it is sufficient to consider $i=n$. Moreover, by \cref{prop:ZSymmetry2}, we may focus on $w_n={\bar q}$. For this specialisation, we obtain
  \begin{equation}
    Z_{2n+1}(\dots,w_n={\bar q}) = (\cdots \otimes \langle \chi({\bar q})|\otimes \langle \varphi|)|\Psi_{2n+1}(\dots,{\bar q},q,1)\rangle.
  \end{equation}
  Using the factorisation property $\langle \chi({\bar q})| = \langle \varphi|\otimes \langle \varphi|$, as well as \cref{lem:PsiReduction}, we find
  \begin{equation}
   Z_{2n+1}(\dots,w_n={\bar q}) \propto (\cdots \otimes \langle \varphi|)|\Psi_{2n-1}(\dots,{\bar q})\rangle(\langle \varphi|\otimes \langle \varphi|)\ket{\zeta}.
  \end{equation}
  The factor $(\langle \varphi|\otimes \langle \varphi|)\ket{\zeta}$ on the right-hand side vanishes.
\end{proof}

By \cref{prop:ZSymmetry,prop:ZSymmetry2,prop:ZSpecialZeros,prop:ZOddExtraZero}, the generalised sum of components is divisible by products of elementary Laurent polynomials. We remove these products by defining, for each $N\geqslant 0$, the function
 \begin{equation}
   \label{eqn:DefY}
    Y_{N}(w_1,\dots,w_n) = \frac{Z_{N}(w_1,\dots,w_n)}{\prod_{i=1}^n [q^{-1/2}w_i]}\times 
    \begin{cases}
     1, & \text{even }N,\\
     \frac{1}{\prod_{i=1}^n [q w_i][q{\bar w}_i]}, & \text{odd }N.
    \end{cases}
  \end{equation}
  Working with $Y_N$, instead of $Z_N$, simplifies the presentation of our results. However, some intermediate expressions still involve $Z_N$. 
  
A Laurent polynomial in $w_1,\dots,w_n$ is called BC$_n$-symmetric if it is a symmetric function in $w_1,\dots,w_n$, and remains invariant under the inversion $w_1$ to $w_1^{-1}$.  Moreover, we call it even if, in addition, it is even with respect to $w_1$. The following statement is an immediate consequence of \cref{prop:ZSymmetry,prop:ZDegree,prop:ZSymmetry2}, and the definition \eqref{eqn:DefY}:
\begin{proposition}
  \label{prop:YBC}
  For each $N\geqslant 2$, $Y_N$ is an even $BC_n$-symmetric Laurent polynomial in $w_1,\dots,w_n$. For each $i=1,\dots,n$, its degree width as a Laurent polynomial in $w_i$ is at most $8(n-1)$.
\end{proposition}

\subsubsection{Reduction}

We now investigate several consequences of the reduction relation given in \cref{lem:PsiReduction}. By analogy, we call the resulting relations for $Y_N$ reduction relations, too. 
\begin{proposition}
  \label{prop:ZBarReduction2}
  For $N\geqslant 2$ and $i=1,\dots,n$, we have
  \begin{multline}
    Y_{N}(\dots,w_i=\ii q^{1/2},\dots)\\ = (-1)^{n+1}\{q^{1/2}\beta\}c_N
\Bigg(\prod_{j\neq i}^n\{q^{3/2}w_i\}^2\{q^{3/2}{\bar w}_i\}^2 \Bigg) Y_{N-2}(\dots,\widehat{w_i},\dots),
    \label{eqn:ZBarReduction2}
  \end{multline}
  where $c_N=1$ for even $N$, and $c_N = \{q^{3/2}\}/\{q^{1/2}\}$ for odd $N$.
\end{proposition}
\begin{proof}
  By \cref{prop:YBC}, it is sufficient to establish the relation for $i=1$. We first consider $Z_N$. Upon specialising $w_1=\ii q^{1/2}$, we find
  \begin{align}
    Z_N(\ii q^{1/2},w_2,\dots,w_n) 
    &= \left(\langle \chi(\ii q^{1/2})|\otimes \cdots\right)|\Psi_N(\ii q^{1/2},-\ii q^{-1/2},\dots)\rangle\\
    &= \left(\langle \chi(\ii q^{1/2})|
    	(\II \otimes \sigma^z)
    	\otimes \cdots\right)|\Psi_N(\ii q^{1/2},\ii q^{-1/2},\dots)\rangle.
  \end{align}
  Here and in the following, $\II$ denotes the identity operator on $\mathbb C^2$ (i.e. the $2\times 2$ identity matrix). Moreover, the second line of this equality follows from the first through an application of \eqref{eqn:PsiReflection}. Next, we apply \cref{lem:PsiReduction} and find
  \begin{multline}
     Z_N(\ii q^{1/2},w_2,\dots,w_n) 
= (-1)^n \ii \left(\langle \chi(\ii q^{1/2})|
(\II \otimes \sigma^z)
\ket{\zeta}\right)\{q^{1/2}\beta\} \\
     \times \prod_{k=2}^n \{q^{3/2}w_k\}^2\{q^{3/2}{\bar w}_k\}^2Z_{N-2}(w_2,\dots,w_n)\times 
     \begin{cases}
      1, & \text{even }N,\\
      \{q^{3/2}\}^2, & \text{odd }N.
     \end{cases}
  \end{multline}
  The overlap on the first line is $\langle \chi(\ii q^{1/2})|
(\II \otimes \sigma^z)
  \ket{\zeta}=-2$. The reduction relation \eqref{eqn:ZBarReduction2} with $i=1$ follows from this equality for $Z_N$ and its combination with \eqref{eqn:DefY}.
\end{proof}

For the next reduction relation, we introduce the abbreviation
\begin{equation}
\label{eqn:Deff}
  f_N(w)=-\frac{[q^2]^2\{q^{1/2}w\}\{q^{3/2}{\bar w}\}[\beta w][\beta q{\bar w}]}{\{q^{1/2}\}^2}\times
  \begin{cases}
    [w][q{\bar w}], &\text{even }N,\\
    [qw][q^2{\bar w}],& \text{odd }N.
  \end{cases}
\end{equation}

\begin{proposition}
  \label{prop:ZBarReduction4}
   For $N\geqslant 4$ and $i,j=1,\dots,n$ with $i<j$, we have
   \begin{multline}
Y_N(\dots,w_i,\dots,w_j = {\bar q}w_i,\dots) = f_N(w_i)\\
 \!\!\times \Bigg(\prod_{k\neq i,j}^n[qw_iw_k][qw_i{\bar w}_k][q^2w_k{\bar w}_i][q^2{\bar w}_i{\bar w}_k]\Bigg)^2Y_{N-4}(\dots,\widehat{w_i},\dots,\widehat{w_j} ,\dots).
     \label{eqn:RedZBar}
   \end{multline}
\end{proposition}
\begin{proof}
  By \cref{prop:YBC}, it is sufficient to consider $i=1,j=2$.
  We focus on even $N$, so that $N=2n$, and abbreviate $Z_{2n}'=Z_{2n}(w_1,{\bar q}w_1,\dots)$. We have
  \begin{align}
   Z_{2n}'&=\langle\chi_{2n}(w_1,{\bar q}w_1,\dots)|\Psi_{2n}(w_1,{\bar w}_1,{\bar q}w_1,q {\bar w}_1,\dots)\rangle\\
   &= \langle\chi_{2n}(w_1,{\bar q}w_1,\dots)|\check R_{2,3}({\bar q}w^2_1)|\Psi_{2n}(w_1,{\bar q}w_1,{\bar w}_1,q {\bar w}_1,\dots)\rangle.
   \nonumber
  \end{align}
  From the first to the second line, we used \cref{lem:Exchange}. The arguments of the vector $|\Psi_{2n}\rangle$ on the right-hand side of this equality allow us to apply \cref{lem:PsiReduction}. We find
  \begin{multline}
    Z_{2n}' = (-1)^{n}[q][q^2][\beta w_1][{\bar q}^2w_1^2][{\bar q}^3w_1^2]\prod_{k=3}^n[q^{2}w_k{\bar w}_1][q^{2}{\bar w}_1{\bar w}_k][qw_1w_k][qw_1{\bar w}_k]\\
    \times \bra{\chi_{2n}(w_1,{\bar q}w_1,\dots)}\check R_{2,3}({\bar q}w^2_1)\left(\ket{\zeta} \otimes \ket{\Psi_{2(n-1)}({\bar w}_1,q {\bar w}_1,\dots)}\right).
    \nonumber
  \end{multline}
  Next, we observe that $\ket{\zeta} = ([q^2]/(2[q]))\check R_{1,2}({\bar q})\ket{\zeta}$ and infer
  \begin{align}
   &\langle\chi_{2n}(w_1,{\bar q}w_1,\dots)|\check R_{2,3}({\bar q}w^2_1)\check R_{1,2}({\bar q})\left(\ket{\zeta} \otimes |\Psi_{2(n-1)}({\bar w}_1,q {\bar w}_1,\dots)\rangle\right) \nonumber \\
   &=\langle\chi_{2n}({\bar q}w_1,w_1,\dots)|
\check R_{2,3}({\bar q}w^2_1)
\check R_{3,4}({\bar q})\left(\ket{\zeta} \otimes |\Psi_{2(n-1)}({\bar w}_1,q {\bar w}_1,\dots)\rangle\right) \nonumber \\
   &= \langle\chi_{2n}({\bar q}w_1,w_1,\dots)|\check R_{2,3}({\bar q}w^2_1)\left(\ket{\zeta} \otimes |\Psi_{2(n-1)}(q {\bar w}_1, {\bar w}_1,\dots)\rangle\right) \nonumber. 
  \end{align}
  Here, the second line follows from the first one through an application of \eqref{eqn:bYBEChi}. The third line results from the application of \cref{lem:Exchange}. The arguments of the vector $|\Psi_{N-2}\rangle$ on the third line allow us to apply, once more, the reduction relation of \cref{lem:PsiReduction}. It leads to
  \begin{multline}
Z_{2n}'=-  \frac12 [q^2]^2[\beta w_1][\beta q{\bar w}_1][{\bar q}^2w_1^2][{\bar q}^3w_1^2]\prod_{k=3}^n\left([qw_iw_k][qw_i{\bar w}_k][q^2w_k{\bar w}_i]q^2 {\bar w}_i{\bar w}_k)]\right)^2 \\
    \times \left(\langle \chi({\bar q}w_1)|\otimes \langle \chi(w_1)|\right)\check R_{2,3}({\bar q}w^2_1)\left(\ket{\zeta} \otimes \ket{\zeta} \right)Z _{2(n-2)}(w_3,\dots,w_n).
    \label{eqn:ZRedIntermediate}
  \end{multline}
  The overlap on the second line is
  \begin{equation}
    \left(\langle \chi({\bar q}w_1)|\otimes \langle \chi(w_1)|\right)\check R_{2,3}({\bar q}w^2_1)\left(\ket{\zeta} \otimes \ket{\zeta}\right) = \frac{2\{q^{1/2}w_1\}[w][q^{1/2}{\bar w}_1]}{\{q^{1/2}\}^2 \{q{\bar w}_1\}}.  \end{equation}
  We insert this expression into \eqref{eqn:ZRedIntermediate} and thus obtain a reduction relation for $Z_{2n}$. Using the definition \eqref{eqn:DefY} leads to  \eqref{eqn:RedZBar} for $i=1$ and $j=2$, which completes the proof for even $N$.
  
  For odd $N$, the proof is similar.
\end{proof}

\subsubsection{Uniqueness}
\label{sec:Uniqueness}

Consider a family of functions $X_N = X_N(w_1,\dots,w_n), \,N\geqslant 0,$ with the following properties:
\begin{enumerate}[label=\textit{(\roman*)}]
  \item $X_0=X_1=1$;
  \item for $N\geqslant 2$, $X_N$ is an even $BC_n$-symmetric Laurent polynomial of degree width at most $8(n-1)$ with respect to $w_1$; 
  \item for $N\geqslant 2$, we have
  \begin{multline}
    X_N(w_1=\ii q^{1/2},w_2,\dots,w_n) = (-1)^{n+1}\{q^{1/2}\beta\}c_N \\
    \times \left(\prod_{i=2}^N \{q^{3/2}w_i\}^2\{q^{3/2}{\bar w}_i\}^2\right) X_{N-2}(w_2,\dots,w_n);
 \end{multline}
  \item for $N\geqslant 4$, we have 
  \begin{multline}
     X_N(w_1, w_2 = \bar q w_1,\dots,w_n)=f_N(w_1) \\
    \times\left(\prod_{k=3}^n[q w_1 w_k][q w_1{\bar w}_k][q^2{\bar w}_1 w_k][q^2{\bar w}_1{\bar w}_k)]\right)^2     X_{N-4}(w_3, \dots,w_n).
  \end{multline}
 \end{enumerate}
The following result will be a key ingredient in characterising the generalised sum of components:
\begin{proposition}
  \label{prop:Uniqueness}
  For each $N\geqslant 0$, we have $X_N(w_1,\dots,w_n) = Y_N(w_1,\dots,w_n)$.
\end{proposition}
\begin{proof}
  By \textit{(i)}, the proof is trivial for $N=0,1$. Hence, we consider $N\geqslant 2$ and, thus, $n\geqslant 1$.
  
   First, by \textit{(ii)}, $X_N(w_1,\dots,w_n)$ is uniquely determined through its values at $8(n-1)+1= 8n-7$ distinct values of $w_1$. The combination of \textit{(ii)} and \textit{(iii)} 
 leads to an expression of $X_N(w_1,\dots,w_n)$ in terms of $X_{N-2}(w_2,\dots,w_N)$ when   \begin{equation}
   w_1=\pm\ii q^{1/2},\,\pm\ii q^{-1/2}.
\end{equation}
Similarly, if $N\geqslant 4$, combining \textit{(ii)} with \textit{(iv)} yields an expression for $X_N(w_1,\dots,w_n)$ in terms of $X_{N-4}(w_2,\dots,w_{i-1},w_{i+1},\dots, w_n)$ when
 \begin{equation}
   w_1 = \pm q w_i,\pm {\bar q} w_i,\pm q {\bar w}_i,\pm {\bar q} {\bar w}_i, 
 \end{equation}
 for each $i=2,\dots,n$. As $X_0$ and $X_1$ are known by \textit{(i)}, it follows by strong induction on $N$ that $X_N(w_1,\dots,w_n)$ is indeed known at $8n-7$ distinct values of $w_1$. (In fact, for $N=2,3$, it is even known at $4$ distinct values of $w_1$, and for $N\geqslant 4$ at $8n-4$ distinct values of $w_1$.)
  
    Second, the equality $X_N(w_1,\dots,w_n) = Y_N(w_1,\dots,w_n)$ follows from checking that $Y_N$ satisfies \textit{(i)}-\textit{(iv)}, too. The definition of $Y_N$  straightforwardly implies \textit{(i)}. The property \textit{(ii)} follows from \cref{prop:YBC}. Moreover, the reduction relations \textit{(iii)} and \textit{(iv)} are a straightforward consequence of \cref{prop:ZBarReduction2} and \cref{prop:ZBarReduction4}, respectively.
\end{proof}

The unique characterisation of a quantity related to an integrable model through symmetry, analyticity, and reduction relations often allows one to derive explicit exact expressions in the form of determinants or Pfaffians. The strategy consists of guessing an expression for the quantity of interest and verifying that it has all the desired properties (see, for example, \cite{izergin:92,kuperberg:02}). In our case, we currently do not have a guess for $Y_N$ with arbitrary $q$ and $\beta$ in terms of an explicit determinant or Pfaffian formula. (We defer the discussion of a few \emph{special} but non-trivial cases, where such a formula can be found, to a separate publication \cite{lienardy:tbp}.) However, in the forthcoming sections, we demonstrate that $Y_N$ can nonetheless be constructed in terms of the partition function of a six-vertex model on a square-grid graph of triangular shape, despite the absence of an explicit formula. In the homogeneous limit, this model and its partition function allow us to relate the sum of components of the special eigenvector to the enumeration of TSASMs.

\section{The six-vertex model and TSASMs}
\label{sec:6VTSASM}

In this section, we discuss the relation between the six-vertex model and a weighted enumeration of TSASMs. We review this enumeration and its corresponding generating function in \cref{sec:TSASM}. In \cref{sec:6VTriangular}, we define the configurations of a six-vertex model on a square-grid graph with triangular shape and relate them to TSASMs. The introduction of weights for the six-vertex model in \cref{sec:Weights} allows us to express the TSASM generating function as a specialisation of the model's partition function with certain boundary conditions. In \cref{sec:Overlap}, we introduce and analyse an overlap that allows us to compute the partition functions. The analysis allows us to prove that the overlap coincides, up to a suitable parameter identification and an elementary factor, with the generalised sum of components.

\subsection{TSASM enumeration}
\label{sec:TSASM}
Let $M\geqslant 1$ be an integer. An alternating sign matrix (ASM) of order $M$ is an $M\times M$ square matrix whose entries are $0,\,1$ or $-1$, such that in each row and column, the non-zero entries alternate in sign, and all row and column sums are equal to one \cite{bressoudbook}. 

A totally-symmetric ASM (TSASM) is an ASM of odd order $2N+1$ that is invariant under all symmetries of the square \cite{bousquet:95,behrend:23}. (There are no TSASMs of even order.)  Since the symmetry group of the square is generated by the reflections with respect to the vertical median and the main diagonal, the entries of a TSASM $A =(A_{i,j})_{i,j=1}^{2N+1}$
satisfy the relations
\begin{equation}
A_{i,2(N+1)-j} = A_{i,j}, \quad  A_{j,i} = A_{i,j},
  \label{eqn:TSASMSymmetry}
\end{equation}
for all $i,j=1,\dots,2N+1$. We denote by $\mathrm{TSASM}(2N+1)$ the set of all TSASMs of order $2N+1$. Two examples of TSASMs of order seven ($N=3$) are\footnote{To improve the readability, we write $+$ and $-$ for the entries $+1$ and $-1$, respectively.}
\begin{equation}
\label{eqn:TSASMExample}
\begin{pmatrix}
  0 & 0 & 0 & + & 0 & 0 & 0  \\
  0 & + & 0 & - & 0 & + & 0  \\
  0 & 0 & 0 & + & 0 & 0 & 0  \\
  + & - & + & - & + & - & +  \\
  0 & 0 & 0 & + & 0 & 0 & 0  \\
  0 & + & 0 & - & 0 & + & 0 \\
  0 & 0 & 0 & + & 0 & 0 & 0  \\
  \end{pmatrix},\quad 
\begin{pmatrix}
  0 & 0 & 0 & + & 0 & 0 & 0  \\
  0 & 0 & + & - & + & 0 & 0  \\
  0 & + & - & + & - & + & 0  \\
  + & - & + & - & + & - & +  \\
  0 & + & - & + & - & + & 0  \\
  0 & 0 & + & - & + & 0 & 0 \\
  0 & 0 & 0 & + & 0 & 0 & 0  
 \end{pmatrix}
  .
\end{equation}
More generally, $\mathrm{TSASM}(2N+1)$ is non-empty for each $N\geqslant 0$:
\begin{example}
\label{example:TSASM}
For each $N\geqslant 0$, the matrix $A^\diamond=(A^\diamond_{i,j})_{i,j=1}^{2N+1}$ with entries
\begin{equation}
  A^\diamond_{i,j} = (-1)^{i+j+N},
\end{equation}
if $|i-j| \leqslant N$ and $|2(N+1)-i-j|\leqslant N$, and $A^\diamond_{i,j}=0$ otherwise, is a TSASM. The second TSASM in \eqref{eqn:TSASMExample} corresponds to $A^\diamond$ with
 $N=3$. 
\end{example}

By \eqref{eqn:TSASMSymmetry}, the entries of the vertical and horizontal median of $A\in \mathrm{TSASM}(2N+1)$ are frozen to
 \begin{equation}
   A_{i,N+1} = (-1)^{i+1}, \quad A_{N+1,j} = (-1)^{j+1},
   \label{eqn:TSASMMedians}
 \end{equation}
 for all $i,j=1,\dots, 2N+1$, irrespectively of the other entries. The medians partition $A$ into four $N\times N$ submatrices, which can be obtained one from another by the symmetry operations. Without loss of generality, we focus on the upper right submatrix. Moreover, combining the two relations \eqref{eqn:TSASMSymmetry} implies
\begin{equation}
  A_{2(N+1)-j,2(N+1)-i}=A_{i,j},
\end{equation}
i.e. the TSASM, and thus the upper right submatrix, is invariant under the reflection with respect to its antidiagonal. Hence, we may reconstruct the full matrix $A$ from the submatrix' entries along and below this antidiagonal. Finally, we recall that the first and last rows and columns of any ASM contain a single entry $+1$, whereas all their other entries are $0$ \cite{bressoudbook}. For the TSASM $A$, this single entry $+1$ is at the row's or column's central position, thanks to \eqref{eqn:TSASMMedians}. It follows that all the entries of the last column of the submatrix are $0$. For a reason to be discussed later, we discard this column for odd $N$, but retain it for even $N$. Hence, we consider the triangular array
\begin{equation}
  \label{eqn:TriangularArray}
  \begin{array}{cccc}
              &             &                       & A_{1+\epsilon,2N+1-\epsilon}\\
              &             & \hspace{-3mm} \tikz[baseline=-2.25mm]{\draw node[rotate=65] (0,0) {$\ddots$};} & \vdots\\
              & A_{N-1,N+3} & \cdots & A_{N-1,2N+1-\epsilon}\\
    A_{N,N+2} & A_{N,N+3} & \cdots & A_{N,2N+1-\epsilon}
  \end{array}
  \vspace{.3cm}
\end{equation}
where $\epsilon$ is defined in \eqref{eqn:Defnnbar}. 
Both the number of rows and the number of columns in this triangular array are $2n$, with $n$ as given in \eqref{eqn:Defnnbar}. Using \eqref{eqn:TSASMSymmetry} and \eqref{eqn:TSASMMedians}, one may reconstruct the full matrix $A$ from the array. \Cref{fig:TSASMTriangles4,fig:TSASMTriangles5} show all triangular arrays for $N=4$ and $N=5$, respectively.
\begin{figure}
\centering
\begin{tikzpicture}
  \draw (-4.5,0) node
  {$\begin{array}{cccc}
         &   &   & 0\\
         &   & 0 & 0\\
         & 0 & 0 & 0\\
       0 & 0 & + & 0
  \end{array}$};
  \draw (-1.5,0) node
  {$\begin{array}{cccc}
         &   &   & 0\\
         &   & + & 0\\
         & 0 & 0 & 0\\
       + & 0 & 0 & 0
  \end{array}$};
  \draw (1.5,0) node
  {$\begin{array}{cccc}
         &   &   & 0\\
         &   & 0 & 0\\
         & - & + & 0\\
       + & 0 & 0 & 0
  \end{array}$};
  \draw (4.5,0) node
  {$\begin{array}{cccc}
         &   &   & 0\\
         &   & 0 & 0\\
         & + & 0 & 0\\
       + & - & + & 0
  \end{array}$};
\end{tikzpicture}
\caption{The triangular arrays corresponding to the elements of $\mathrm{TSASM(9)}$.}
\label{fig:TSASMTriangles4}
\end{figure}
\begin{figure}
\centering
\begin{tikzpicture}
  \draw (-4.5,0) node
  {$
  \begin{array}{cccc}
  &  &  & 0 \\
  &  & 0 & 0 \\
  & 0 & 0 & + \\
 0 & 0 & 0 & 0 \\
\end{array}$
};
 \draw (-1.5,0) node
  {$
  \begin{array}{cccc}
   &   &   & 0 \\
   &   & 0 & 0 \\
   & 0 & + & 0 \\
 0 & 0 & - & + \\
\end{array}$
};
 \draw (1.5,0) node
  {$
  \begin{array}{cccc}
   &   &   & + \\
   &   & 0 & 0 \\
   & + & 0 & 0 \\
 0 & 0 & 0 & 0 \\
\end{array}$
};
 \draw (4.5,0) node
  {$
 \begin{array}{cccc}
   &   &   & 0 \\
   &   & + & 0 \\
   & + & - & + \\
 0 & 0 & 0 & 0 \\
\end{array}$
};
\begin{scope}[yshift=-2.25cm]
  \draw (-4.5,0) node
  {$
 \begin{array}{cccc}
   &   &   & 0 \\
   &   & + & 0 \\
   & + & 0 & 0 \\
 0 & 0 & - & + \\
\end{array}$
};
 \draw (-1.5,0) node
  {$
  \begin{array}{cccc}
   &   &   & 0 \\
   &   & - & + \\
   & + & 0 & 0 \\
 0 & 0 & 0 & 0 \\
\end{array}$
};
 \draw (1.5,0) node
  {$
  \begin{array}{cccc}
   &   &   & + \\
   &   & 0 & 0 \\
   & 0 & 0 & 0 \\
 - & + & 0 & 0 \\
\end{array}$
};
 \draw (4.5,0) node
  {$
 \begin{array}{cccc}
   &   &   & 0 \\
   &   & + & 0 \\
   & 0 & - & + \\
 - & + & 0 & 0 \\
\end{array}$
};
\end{scope}
\begin{scope}[yshift=-4.5cm]
  \draw (-4.5,0) node
  {$
 \begin{array}{cccc}
   &   &   & 0 \\
   &   & + & 0 \\
   & 0 & 0 & 0 \\
 - & + & - & + \\
\end{array}$
};
 \draw (-1.5,0) node
  {$
  \begin{array}{cccc}
   &   &   & 0 \\
   &   & 0 & 0 \\
   & + & 0 & 0 \\
 - & 0 & 0 & + \\
\end{array}$
};
 \draw (1.5,0) node
  {$
  \begin{array}{cccc}
   &   &   & 0 \\
   &   & - & + \\
   & 0 & 0 & 0 \\
 - & + & 0 & 0 \\
\end{array}$
};
 \draw (4.5,0) node
  {$
 \begin{array}{cccc}
   &   &   & 0 \\
   &   & 0 & 0 \\
   & - & 0 & + \\
 - & + & 0 & 0 \\
\end{array}$
};
\end{scope}
\begin{scope}[yshift=-6.75cm]
  \draw (0,0) node
  {$
 \begin{array}{cccc}
   &   &   & 0 \\
   &   & 0 & 0 \\
   & - & + & 0 \\
 - & + & - & + \\
\end{array}$
};
\end{scope}
\end{tikzpicture}
\caption{The triangular arrays corresponding to the elements of $\mathrm{TSASM}(11)$.}
\label{fig:TSASMTriangles5}
\end{figure}

For each $A\in\mathrm{TSASM}(2N+1)$, we define $\mu(A)$ and $\nu(A)$ as the numbers of the non-zero entries along and below the diagonal of its triangular array:
\begin{equation}
  \mu(A) = \sum_{i=1+\epsilon}^N |A_{i,2(N+1)-i}|,\quad \nu(A) = \sum_{1+\epsilon \leqslant j < i \leqslant N}|A_{i,2(N+1)-j}|.
\end{equation}
We define the TSASM generating function as
\begin{equation}
  \label{eqn:TSASMGF}
  A_{\mathrm{TS}}(2N+1;t,\tau) = \sum_{A \in \mathrm{TSASM}(2N+1)}t^{\mu(A)}\tau^{\nu(A)},
\end{equation}
where $t$ and $\tau$ are complex numbers. 
\begin{example}
  The first generating functions are $A_{\mathrm{TS}}(1;t,\tau) =1$, $A_{\mathrm{TS}}(3;t,\tau) = 1$, and
  \begin{align}
      A_{\mathrm{TS}}(5;t,\tau) &= t,\\
      A_{\mathrm{TS}}(7;t,\tau) &= t(1+\tau),\\
      A_{\mathrm{TS}}(9;t,\tau) &= \tau + t^2 (1+\tau + \tau^2),\\
      A_{\mathrm{TS}}(11;t,\tau) & =\tau(1+\tau^2)+t^2(1+3\tau+4\tau^2+2\tau^3+\tau^4).
  \end{align}
\end{example}

\subsection{Six-vertex configurations and TSASMs}
\label{sec:6VTriangular}
We now introduce a six-vertex model on a square-grid graph of triangular shape, whose configurations are in one-to-one correspondence with the TSASM triangular arrays defined above. To this end, we adapt the considerations of Behrend, Fischer, and Koutschan \cite[Section 5.1]{behrend:23} to our setting. 

For $n\geqslant 1$, we consider the following square-grid graph:
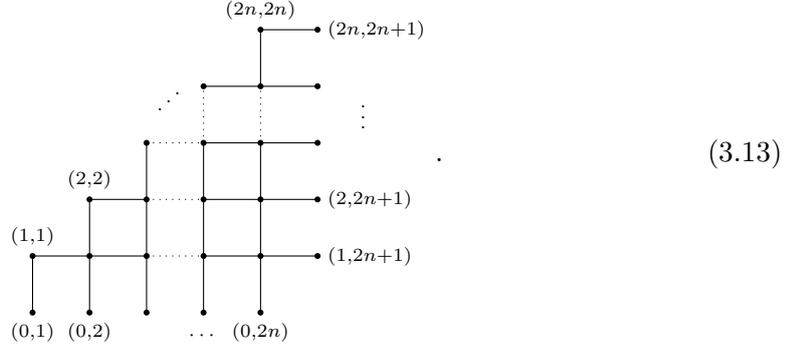
\begin{equation}
\label{eqn:TriangularGraph}
\begin{tikzpicture}[scale=.75,baseline=2cm]
    \foreach \i in {1,2,3}
    {
      \draw (\i cm,0) -- (\i cm,\i cm) -- (3,\i cm);
      \draw[dotted] (3,\i cm) -- (4,\i cm);
      \draw (4 cm,\i cm) -- (6,\i cm);
    }
    
    \foreach \i in {4,5}
    {
      \draw (\i cm,0) -- (\i cm,3);
      \draw[dotted] (\i cm,3) -- (\i cm,4);
      \draw (\i cm,4) -- (\i cm, \i cm) -- (6,\i cm);
    }
    
    \foreach \i in {1,2,...,5}
    {
      \foreach \j in {1,...,\i}
      {
        \fill (\i,\j) circle (1.5pt);
      }
    }
    
    \foreach \i in {1,2,...,5}
    {
      \fill (\i,0) circle (1.5pt);
      \fill (6,\i) circle (1.5pt);
    }
    
    \draw (1,0) node[below] {\tiny ({0},{1})};
    \draw (2,0) node[below] {\tiny ({0},{2})};
    \draw (5,0) node[below] {\tiny ({0},{$2n$})};

    \draw (1,1) node[above] {\tiny ({1},{1})};
    \draw (2,2) node[above] {\tiny ({2},{2})};
 
    \draw (5,5) node[above] {\tiny ($2n$,$2n$)};
     
    \draw (6,1) node[right] {\tiny ($1$,$2n{+}1$)};
    \draw (6,2) node[right] {\tiny ($2$,$2n{+}1$)};
    \draw (6,5) node[right] {\tiny ($2n$,$2n{+}1$)};
    
    \draw (4,-.4) node {\tiny $\cdots$};
    \draw (6.8,3.6) node {\tiny $\vdots$};
    \draw (3.25,3.75) node[rotate=85] {\tiny $\ddots$};
     
  \end{tikzpicture}.
\end{equation}
It has four types of vertices: the bulk vertices $(i,j),\, 1\leqslant i< j\leqslant 2n$, of degree $4$, the corner vertices $(i,i),\,1\leqslant i \leqslant 2n$, of degree $2$, and the bottom and right boundary vertices,  $(0,j),\,1\leqslant j \leqslant 2n$, and $(i,2n+1),\,1\leqslant i \leqslant 2n$, respectively, of degree $1$. We call an edge a bottom or right edge if it is incident to a bottom or right boundary vertex, respectively. 

We consider a six-vertex model on the graph \eqref{eqn:TriangularGraph} with boundary conditions at the bottom encoded in $\alpha_1,\alpha_2,\dots, \alpha_{2n}\in \{\uparrow,\downarrow\}$. A configuration of this model is an orientation of the edges that obeys the following conditions: \textit{(i)} Two edges are directed toward and two edges are directed away from each internal vertex (ice rule); \textit{(ii)} the right boundary edges are oriented leftward; \textit{(iii)} for each $1\leqslant i \leqslant 2n$, the bottom boundary edge incident to the boundary vertex $(i,0)$ possesses the orientation $\alpha_i$. We denote the set of these configurations by $\mathrm{6V}^{\bm\alpha}_n$, where $\bm \alpha = \alpha_1\alpha_2 \cdots \alpha_{2n}$. 

For each configuration $C\in \mathrm{6V}^{\bm\alpha}_n$ and each vertex $(i,j)$, the local configuration $C_{ij}$ at $(i,j)$ is the orientation of the edges that are incident to $(i,j)$. There are six possible local configurations at a bulk vertex (hence the model's name),
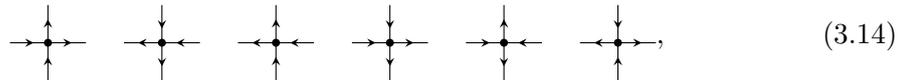
\begin{equation}
  \begin{tikzpicture}[baseline=0cm]
    \foreach \i in {1,2,...,6}
    {
      \drawvertex{\i}{1.5*\i}{0} 
    }
    \end{tikzpicture},
\end{equation}
and four possible local configurations at a corner vertex
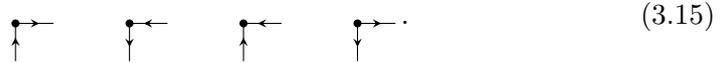
\begin{equation}
   \begin{tikzpicture}[baseline=0cm]
     \draw[white] (0,0)--(0,.5);
     \foreach \i in {1,2,3,4}
     {       
       \drawcornervertex{\i}{1.5*\i}{0} 
     }    
   \end{tikzpicture}
   \,.
\end{equation}
Finally, the possible local configurations at the boundary vertices are
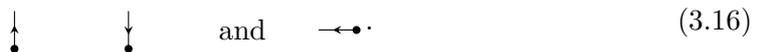
\begin{equation}
  \begin{tikzpicture}[baseline=.25cm]
     \foreach \i in {0,1}
     {       
       \fill (1.5*\i,0) circle (1.5pt);
     }
     \draw[postaction={on each segment={mid arrow}}] (0,0)--(0,.5);
     \draw[xshift=1.5cm,postaction={on each segment={mid arrow}}] (0,.5)--(0,0);
     
     \draw (3,.25) node {$\mathrm{and}$};
 
     \fill (4.5,.25) circle (1.5pt);
     \draw[xshift=4.5cm,postaction={on each segment={mid arrow}}] (0,.25)--(-.5,.25);
\end{tikzpicture}
\,.
\end{equation}
The boundary conditions may exclude some of these local configurations for boundary vertices or vertices that are adjacent to a boundary vertex.

We are particularly interested in the cases where the boundary condition $\bm \alpha$ at the bottom is fixed to a sequence of alternating orientations:
\begin{equation}
  \bm \alpha_+ =\, \uparrow\downarrow \cdots \uparrow\downarrow, \quad \bm \alpha_- = \,\downarrow\uparrow\cdots \downarrow\uparrow.
\end{equation}
To simplify the notation, we write $\mathrm{6V}^\pm_n$ for $\mathrm{6V}^{\bm{\alpha}_\pm}_n$. \Cref{fig:6VC1,fig:6VC2} illustrate the elements of $\mathrm{6V}_n^-$ and $\mathrm{6V}_n^+$, respectively, for $n=2$.
\begin{figure}[h]
  \begin{tikzpicture} 
    \confsixvertex{{4,1,2,3},{0,3,3,3},{0,0,3,3},{0,0,0,3}}{0}{0}
    \confsixvertex{{2,3,2,3},{0,3,2,3},{0,0,2,3},{0,0,0,3}}{3}{0}
    \confsixvertex{{2,3,2,3},{0,1,2,3},{0,0,3,3},{0,0,0,3}}{6}{0}
	\confsixvertex{{2,1,2,3},{0,2,3,3},{0,0,3,3},{0,0,0,3}}{9}{0} 	
  \end{tikzpicture}
    \caption{Illustration of the elements of $\mathrm{6V}_n^-$ for $n=2$.}
    \label{fig:6VC1}
\end{figure}
\begin{figure}[ht]
  \begin{tikzpicture}  
    \confsixvertex{{3,2,3,2},{0,4,1,2},{0,0,3,3},{0,0,0,3}}{0}{0}
    \confsixvertex{{3,2,1,2},{0,4,2,3},{0,0,3,3},{0,0,0,3}}{3}{0}
    \confsixvertex{{3,2,3,2},{0,2,3,2},{0,0,3,2},{0,0,0,2}}{6}{0}
	\confsixvertex{{3,2,3,2},{0,2,1,2},{0,0,2,3},{0,0,0,3}}{9}{0}

    \confsixvertex{{3,2,1,2},{0,2,2,3},{0,0,2,3},{0,0,0,3}}{0}{-3}
    \confsixvertex{{3,2,3,2},{0,2,3,2},{0,0,1,2},{0,0,0,3}}{3}{-3}
    \confsixvertex{{1,2,3,2},{0,3,3,2},{0,0,3,2},{0,0,0,2}}{6}{-3}
	\confsixvertex{{1,2,3,2},{0,3,1,2},{0,0,2,3},{0,0,0,3}}{9}{-3}

	\confsixvertex{{1,2,1,2},{0,3,2,3},{0,0,2,3},{0,0,0,3}}{0}{-6}
    \confsixvertex{{1,4,1,2},{0,2,3,3},{0,0,3,3},{0,0,0,3}}{3}{-6}
	\confsixvertex{{1,2,3,2},{0,3,3,2},{0,0,1,2},{0,0,0,3}}{6}{-6}
	\confsixvertex{{1,2,3,2},{0,1,1,2},{0,0,3,3},{0,0,0,3}}{9}{-6}

	\confsixvertex{{1,2,1,2},{0,1,2,3},{0,0,3,3},{0,0,0,3}}{4.5}{-9}

  \end{tikzpicture}

  \caption{Illustration of the elements of $\mathrm{6V}_n^+$ for $n=2$.}
    \label{fig:6VC2}
\end{figure}

We now describe the connection between the six-vertex model configurations and TSASMs. To this end, we recall that, for each integer $M\geqslant 1$, there is an explicit bijection between the set of ASMs of order $M$ and the configurations of the six-vertex model with domain-wall boundary conditions on a square-grid graph associated with an $M\times M$ square \cite{elkies:92}. We set $M=2N+1$ and consider the restriction of this bijection to $\mathrm{TSASM}(2N+1)$. For each $A\in \mathrm{TSASM}(2N+1)$, the property \eqref{eqn:TSASMMedians} implies that the corresponding six-vertex configuration is fixed along its horizontal and vertical median to an alternating sequence of the local configurations
\begin{equation}
  \tikz[baseline=-.1cm]{\drawvertex{5}{0}{0}} \quad \text{and} \quad \tikz[baseline=-.1cm]{\drawvertex{6}{0}{0}}.
\end{equation}
The medians divide the configuration into four subconfigurations corresponding to $N\times N$ squares. We focus on the upper-right square and note that the symmetries \eqref{eqn:TSASMSymmetry} imply that the local configurations
\begin{equation}
    \tikz[baseline=-.1cm]{\drawvertex{1}{0}{0}} \quad \text{and} \quad \tikz[baseline=-.1cm]{\drawvertex{2}{0}{0}}
\end{equation}
cannot occur at the vertices on the antidiagonal of this upper-right square \cite[Section 5.1]{behrend:23}. We consider the part on or below its antidiagonal, removing the last row for odd $N$, but keeping it for even $N$. This restriction provides a bijection between $\mathrm{TSASM}(2N+1)$ and $\mathrm{6V}^-_n$ if $N$ is even, and $\mathrm{6V}^+_n$ if $N$ is odd. The matrix $A\in \mathrm{TSASM}(2N+1)$ corresponding to $C\in \mathrm{6V}^\pm_n$ is defined through its triangular array \eqref{eqn:TriangularArray}, whose entries are given by
\begin{equation}
  \label{eqn:6VTSASM}
  A_{i,j} =
  \begin{cases}
    1, & C_{N+1-i,j-(N+1)} = 
	\begin{tikzpicture}[baseline=-.1cm]\drawvertex{5}{0}{0}\end{tikzpicture}\ \mathrm{or}\
	\begin{tikzpicture}[baseline=-.1cm]\drawcornervertex{2}{0}{0}\end{tikzpicture}, \\[3mm]
      -1, & C_{N+1-i,j-(N+1)} = 
	\begin{tikzpicture}[baseline=-.1cm]\drawvertex{6}{0}{0}\end{tikzpicture} \ \mathrm{or}\  
	\begin{tikzpicture}[baseline=-.1cm]\drawcornervertex{1}{0}{0}\end{tikzpicture}, \\[3mm]
      0,& C_{N+1-i,j-(N+1)} = 
	\begin{tikzpicture}[baseline=-.1cm]\drawvertex{1}{0}{0}\end{tikzpicture},\,
	\begin{tikzpicture}[baseline=-.1cm]\drawvertex{2}{0}{0}\end{tikzpicture},\,
	\begin{tikzpicture}[baseline=-.1cm]\drawvertex{3}{0}{0}\end{tikzpicture},\,
	\begin{tikzpicture}[baseline=-.1cm]\drawvertex{4}{0}{0}\end{tikzpicture},\,
	\begin{tikzpicture}[baseline=-.1cm]\drawcornervertex{3}{0}{0}\end{tikzpicture}\ \mathrm{or}\ 
	\begin{tikzpicture}[baseline=-.1cm]\drawcornervertex{4}{0}{0}\end{tikzpicture},
  \end{cases}
\end{equation}
where $i=1+\epsilon,\dots N$ and $j=2(N+1)-i,\dots, 2N+1-\epsilon$. Note that applying this mapping to the configurations shown in \cref{fig:6VC1,fig:6VC2} yields the triangular arrays shown in \cref{fig:TSASMTriangles4,fig:TSASMTriangles5}, respectively.

We now exploit the bijection to characterise the number of non-zero entries along and below the diagonal of the triangular domain associated with each $\mathrm{TSASM}$.

\begin{proposition}\label{prop:boundsDegree}
 For each $N\geqslant 0$ and $A\in \mathrm{TSASM}(2N+1)$, we have $\mu(A) \leqslant n$, $\mu(A) \equiv n \mod 2$, and $(n-\mu(A))/2 \leqslant \nu(A) \leqslant n(n'-1)$.
\end{proposition}

\begin{proof}
  Let $C\in \mathrm{6V}^\pm_n$ be the six-vertex model configuration associated with $A$ through the bijection \eqref{eqn:6VTSASM}, where the superscript $+$ corresponds to odd $N$ and $-$ corresponds to even $N$.

We write
  \begin{equation}
\gamma \left(\begin{tikzpicture}[baseline=-.25cm]\drawcornervertex{1}{0}{0}\end{tikzpicture}\right), 
\gamma \left(\begin{tikzpicture}[baseline=-.25cm]\drawcornervertex{2}{0}{0}\end{tikzpicture}\right), 
\gamma \left(\begin{tikzpicture}[baseline=-.25cm]\drawcornervertex{3}{0}{0}\end{tikzpicture}\right),
\gamma \left(\begin{tikzpicture}[baseline=-.25cm]\drawcornervertex{4}{0}{0}\end{tikzpicture}\right)
  \end{equation}
for the total number of occurrences in $C$ of the four possible local configurations at the corner vertices. Clearly, their sum equals the number of corner vertices:
  \begin{equation}
   \label{eqn:GammaIdentity1}
\gamma \left(\begin{tikzpicture}[baseline=-.25cm]\drawcornervertex{1}{0}{0}\end{tikzpicture}\right)+
\gamma \left(\begin{tikzpicture}[baseline=-.25cm]\drawcornervertex{2}{0}{0}\end{tikzpicture}\right)+
\gamma \left(\begin{tikzpicture}[baseline=-.25cm]\drawcornervertex{3}{0}{0}\end{tikzpicture}\right)+
\gamma \left(\begin{tikzpicture}[baseline=-.25cm]\drawcornervertex{4}{0}{0}\end{tikzpicture}\right)= 2n.
  \end{equation}

Recall that for any oriented graph, the sum of the in-degrees over all vertices equals the sum of the out-degrees over all vertices, and both are equal to the number of edges. Any local configuration at the bulk vertices and the local configurations
  \begin{equation}
    \begin{tikzpicture}[baseline=-.25cm]\drawcornervertex{1}{0}{0}\end{tikzpicture},
    \quad \begin{tikzpicture}[baseline=-.25cm]\drawcornervertex{2}{0}{0}\end{tikzpicture},
  \end{equation}
  at the corner vertices have equal in-degree and out-degree. Therefore, their contributions to both sides of the equality are the same and can thus be removed. Considering the contributions of the remaining corner and boundary vertices of $C$, we obtain the identity
  \begin{equation}
    2\left(\gamma\left(\begin{tikzpicture}[baseline=-.25cm]\drawcornervertex{3}{0}{0}\end{tikzpicture}
    \right)-\gamma \left(\begin{tikzpicture}[baseline=-.25cm]\drawcornervertex{4}{0}{0}\end{tikzpicture}\right)
    \right) = 
    2n.
    \label{eqn:GammaIdentity2}
  \end{equation}

From \eqref{eqn:6VTSASM}, we infer 
  \begin{equation}
    \mu(A) = \gamma \left(\begin{tikzpicture}[baseline=-.25cm]\drawcornervertex{1}{0}{0}\end{tikzpicture}\right)+ \gamma \left(\begin{tikzpicture}[baseline=-.25cm]\drawcornervertex{2}{0}{0}\end{tikzpicture}\right).
  \end{equation}
  We combine this identification with the identities \eqref{eqn:GammaIdentity1} and \eqref{eqn:GammaIdentity2}, which leads to
  \begin{equation}
    \mu(A) = n - 2 \gamma \left(\begin{tikzpicture}[baseline=-.25cm]\drawcornervertex{4}{0}{0}\end{tikzpicture}\right),
\label{eqn:GammaIdentity3}
  \end{equation}
  This equality straightforwardly implies $\mu(A) \leqslant n$ and $\mu(A) \equiv n \mod 2$.


To obtain a lower bound on $\nu(A)$, we focus on the horizontal lines formed by the horizontal edges of the graph \eqref{eqn:TriangularGraph}. The boundary conditions impose that their right boundary edges are oriented leftward. As for the leftmost edge, it is oriented rightward if the corner vertex has the configuration
  \begin{equation}
    \begin{tikzpicture}[baseline=-.25cm]\drawcornervertex{1}{0}{0}\end{tikzpicture},
    \quad \begin{tikzpicture}[baseline=-.25cm]\drawcornervertex{4}{0}{0}\end{tikzpicture}.
  \end{equation}
Therefore, a necessary condition for the ice rule to be respected is that
\begin{equation}
\nu(A) \geqslant \gamma\left(\begin{tikzpicture}[baseline=-.25cm]\drawcornervertex{1}{0}{0}\end{tikzpicture}\right) +
\gamma\left(\begin{tikzpicture}[baseline=-.25cm]\drawcornervertex{4}{0}{0}\end{tikzpicture}\right) \geqslant \gamma\left(\begin{tikzpicture}[baseline=-.25cm]\drawcornervertex{4}{0}{0}\end{tikzpicture}\right).
\end{equation}
Combining this inequality with \eqref{eqn:GammaIdentity3}, we obtain $ \frac{n-\mu(A)}{2} \leqslant \nu(A)$.


Finally, to find the upper bound on $\nu(A)$, we define $\nu_i(A)$ as the number of non-zero entries below the diagonal in the $i$-th column of the triangular array associated with $A$, for each $i=1,\dots,2n$. Clearly, we have
\begin{equation}
  \nu(A) = \sum_{i=1}^{2n}\nu_i(A).
  \label{eqn:NuSum}
\end{equation}
An obvious upper bound is $\nu_i(A) \leqslant i-1$. This estimate can be sharpened using the fact that the number of non-zero entries in the $j$-th column of $A$ is at most $2\min(j,2N+2-j)-1$ \cite[Lemma 2.1]{brualdi:13}. Setting $j=N+1+i$ and exploiting the symmetries \eqref{eqn:TSASMSymmetry}, we infer $\nu_i(A) \leqslant N-i.$ Combining the two bounds, we find
\begin{equation}
\nu_i(A) \leqslant \min(i-1,N-i).
\label{eqn:GammaIdentity4}
\end{equation}
Hence, by \eqref{eqn:NuSum}, we obtain
\begin{equation}
\nu(A) \leqslant\sum_{i=1}^{2n} \min(i-1,N-i) = n(n'-1),
\end{equation}
which concludes the proof.
\end{proof}

Note that this proposition bounds the degrees of $A_{\mathrm{TS}}(2N+1;t,\tau)$ as a polynomial in $t$ and $\tau$. 
The maximal degrees in both $\tau$ and $t$ are attained by the TSASM $A^\diamond$, introduced in \cref{example:TSASM}, for which we have $\mu(A^\diamond)=n$ and $\nu(A^\diamond) = n(n'-1)$, for each $N\geqslant 0$. Moreover, we immediately infer the parity of $A_{\mathrm{TS}}(2N+1;t,\tau)$ as a function of $t$:
\begin{corollary}
  \label{cor:ParityATS}
  $A_{\mathrm{TS}}(2N+1;-t,\tau)=(-1)^nA_{\mathrm{TS}}(2N+1;t,\tau)$.
\end{corollary} 
This parity property implies $A_{\mathrm{TS}}(8k+5;0,\tau) = A_{\mathrm{TS}}(8k+7;0,\tau)=0$ for each integer $k\geqslant 0$. Hence, there are no TSASMs of order $8k+5$ or $8k+7$ without non-zero entries along their (anti)diagonals, which is a well-known property \cite{okada:06}.

\subsection{Weights and the partition function}
\label{sec:Weights}

To compute the TSASM generating function using the bijection defined above, we assign a weight $\WW(C)$ to each configuration $C\in \mathrm{6V}^{\bm \alpha}_n$. 

\subsubsection{Definition} The weight depends on $2n$ complex numbers $z_1,\dots,z_{2n}\in \mathbb C^\times$, called the inhomogeneity parameters, and on two parameters $t\in \mathbb C,\,q\in \mathbb C^\times$ with the restriction $q^4\neq 1$. The inhomogeneity parameters are assigned to the horizontal and vertical lines of the graph \eqref{eqn:TriangularGraph} as follows:
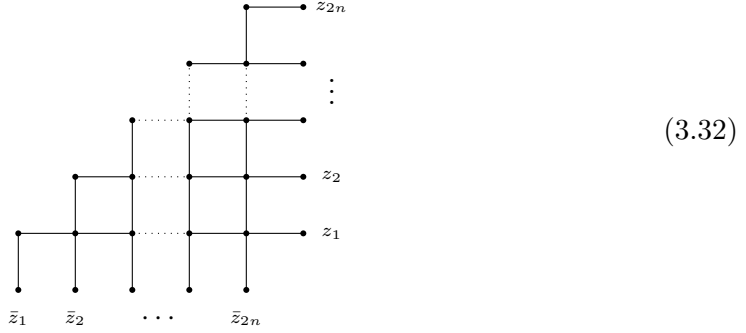
\begin{equation}
  \begin{tikzpicture}[scale=.75,baseline=2cm]
    \foreach \i in {1,2,3}
    {
      \draw (\i cm,0) -- (\i cm,\i cm) -- (3,\i cm);
      \draw[dotted] (3,\i cm) -- (4,\i cm);
      \draw (4 cm,\i cm) -- (6,\i cm);
    }
    
    \foreach \i in {4,5}
    {
      \draw (\i cm,0) -- (\i cm,3);
      \draw[dotted] (\i cm,3) -- (\i cm,4);
      \draw (\i cm,4) -- (\i cm, \i cm) -- (6,\i cm);
    }
    
    \foreach \i in {1,2,...,5}
    {
      \foreach \j in {1,...,\i}
      {
        \fill (\i,\j) circle (1.5pt);
      }
    }
    
    \foreach \i in {1,2,...,5}
    {
      \fill (\i,0) circle (1.5pt);
      \fill (6,\i) circle (1.5pt);
    }
  
%
        
    \draw (1,-.5) node {\tiny $\bar z_1$};
    \draw (2,-.5) node {\tiny $\bar z_2$};
    \draw (3.5,-.5) node {$\cdots$};
    \draw (5,-.5) node {\tiny $\bar z_{2n}$};
    
    \draw (6.5,1) node {\tiny $z_1$};
    \draw (6.5,2) node {\tiny $z_2$};
    \draw (6.5,3.65) node {$\vdots$};
    \draw (6.5,5) node {\tiny $z_{2n}$};
  
  \end{tikzpicture}
\end{equation}

We first define the weights of the local configurations. The weights of the six local configurations at the bulk vertices are
\begin{subequations}
\label{eqn:BulkWeights}
\begin{align}
  \WW\Biggl(
  \tikz[baseline=-.1cm]{
    \drawvertex{1}{0}{0}
    \draw (0.7,0) node {\tiny $z_i$};
    \draw (0,-0.7) node {\tiny $\bar z_j$};
  }\hspace{-.15cm}
  \Biggr)& =\WW\Biggl(
  \tikz[baseline=-.1cm]{
    \drawvertex{2}{0}{0}
    \draw (0.7,0) node {\tiny $z_i$};
    \draw (0,-0.7) node {\tiny $\bar z_j$};
  }\hspace{-.15cm}
  \Biggr)=[q\bar z_i\bar z_j],\\
  \WW\Biggl(
  \tikz[baseline=-.1cm]{
    \drawvertex{3}{0}{0}
    \draw (0.7,0) node {\tiny $z_i$};
    \draw (0,-0.7) node {\tiny $\bar z_j$};
  }\hspace{-.15cm}
  \Biggr)&=\WW\Biggl(
  \tikz[baseline=-.1cm]{
    \drawvertex{4}{0}{0}
    \draw (0.7,0) node {\tiny $z_i$};
    \draw (0,-0.7) node {\tiny $\bar z_j$};
  }\hspace{-.15cm}
  \Biggr)=[q z_iz_j],
  \\
  \WW\Biggl(
  \tikz[baseline=-.1cm]{
    \drawvertex{5}{0}{0}
    \draw (0.7,0) node {\tiny $z_i$};
    \draw (0,-0.7) node {\tiny $\bar z_j$};
  }\hspace{-.15cm}
  \Biggr)&=\WW\Biggl(
  \tikz[baseline=-.1cm]{
    \drawvertex{6}{0}{0}
    \draw (0.7,0) node {\tiny $z_i$};
    \draw (0,-0.7) node {\tiny $\bar z_j$};
  }\hspace{-.15cm}
  \Biggr)=-[q^2],
\end{align}
\end{subequations}%
where we use the notation introduced in \eqref{eqn:NotationBrackets}. Moreover, it is often convenient to extend the definition of the weight to local configurations that do not respect the ice rule. We define their weight as zero. The weights of the local configurations at the corner vertices are
\begin{subequations}
\label{eqn:CornerWeights}
\begin{align}
   \WW  \Biggl(
  \begin{tikzpicture}[baseline=-.4cm]
    \fill (0,0) circle (1.5pt);
    \draw[postaction={on each segment={mid arrow}}] (0,-.5) -- (0,0) -- (.5,0);
    \draw (0,-.75) node {\tiny $\bar z_i$};
    \draw (.75,0) node {\tiny $z_i$};
  \end{tikzpicture}
  \hspace{-.1cm}
  \Biggr)
  &=
  \WW  \Biggl(
  \begin{tikzpicture}[baseline=-.4cm]
    \fill (0,0) circle (1.5pt);
    \draw[postaction={on each segment={mid arrow}}] (.5,0) -- (0,0) -- (0,-.5);
    \draw (0,-.75) node {\tiny $\bar z_i$};
    \draw (.75,0) node {\tiny $z_i$};
  \end{tikzpicture}
  \hspace{-.1cm}
  \Biggr)=t,\\
   \WW  \Biggl(
  \begin{tikzpicture}[baseline=-.4cm]
    \fill (0,0) circle (1.5pt);
    \draw[postaction={on each segment={mid arrow}}] (0,-.5) -- (0,0);
    \draw[postaction={on each segment={mid arrow}}] (.5,0) -- (0,0);
     \draw (0,-.75) node {\tiny $\bar z_i$};
    \draw (.75,0) node {\tiny $z_i$};
  \end{tikzpicture}
  \hspace{-.1cm}
  \Biggr)
  &=
  \WW  \Biggl(
  \begin{tikzpicture}[baseline=-.4cm]
    \fill (0,0) circle (1.5pt);
    \draw[postaction={on each segment={mid arrow}}] (0,0) -- (0,-.5);
    \draw[postaction={on each segment={mid arrow}}] (0,0) -- (.5,0);
    \draw (0,-.75) node {\tiny $\bar z_i$};
    \draw (.75,0) node {\tiny $z_i$};
  \end{tikzpicture}
  \hspace{-.1cm}
  \Biggr)=\{q^{1/2}z_i\}/\{q^{1/2}\}.
\end{align}
\end{subequations}%
Finally, no weight is assigned to the local configuration at the boundary vertices.

The weight of a configuration $C\in \mathrm{6V}^{\bm \alpha}$ is the product of the weights assigned to the local configurations of all its bulk or corner vertices. The partition function of the six-vertex model is obtained by summing the weights of all configurations:
\begin{equation}
  \label{eqn:PartitionFunction6V}
  \ZZ_{n}^{\bm \alpha}(z_1,\dots,z_{2n})=\sum_{C \in \mathrm{6V}^{\bm \alpha}_n} \WW(C)
 = \sum_{C \in \mathrm{6V}^{\bm \alpha}_n} \prod_{1\leqslant i \leqslant j \leqslant 2n}\WW (C_{ij}).
\end{equation}
In the following, we write $\mathbb Z_{n}^{\pm}(z_1,\dots,z_{2n})$ for $\mathbb Z_{n}^{\bm \alpha_{\pm}}(z_1,\dots,z_{2n})$. Moreover, when evaluating the weight of (a subgraph of) \eqref{eqn:TriangularGraph} with unoriented edges, the summation over all possible orientations of these edges is implied.

\begin{example}
  \label{ex:Zpm1}
  For $n=1$, we obtain
  \begin{align}
    \ZZ_n^+(z_1,z_2) &= \WW
    \Biggl(
    \begin{tikzpicture}[baseline=-.2cm]
      \draw[postaction={on each segment={mid arrow}}] (0,-.5) -- (0,0);
      \draw[postaction={on each segment={mid arrow}}] (.5,0) -- (.5,-.5);
      \draw[postaction={on each segment={mid arrow}}] (1,0) -- (.5,0);
      \draw[postaction={on each segment={mid arrow}}] (1,.5) -- (.5,.5);
      \fill (0,0) circle (1.5pt);
      \fill (.5,0) circle (1.5pt);
      \fill (.5,.5) circle (1.5pt);      
      \fill (1,.5) circle (1.5pt);
      \fill (1,0) circle (1.5pt);
      \fill (0,-.5) circle (1.5pt);
      \fill (.5,-.5) circle (1.5pt);
      \draw (0,0) -- (.5,0)--(.5,.5);
      \draw (0,-.5) node[below] {\tiny  $\bar z_1$};
      \draw (.5,-.5) node[below] {\tiny  $\bar z_2$};
      \draw (1,.5) node[right] {\tiny $z_2$};
      \draw (1,0) node[right] {\tiny $z_1$};
    \end{tikzpicture}
    \Biggr)
    =
    \WW
    \Biggl(
    \begin{tikzpicture}[baseline=-.2cm]
      \draw[postaction={on each segment={mid arrow}}] (0,-.5) -- (0,0);
      \draw[postaction={on each segment={mid arrow}}] (.5,0) -- (.5,-.5);
      \draw[postaction={on each segment={mid arrow}}] (1,0) -- (.5,0);
      \draw[postaction={on each segment={mid arrow}}] (1,.5) -- (.5,.5);
      \fill (0,0) circle (1.5pt);
      \fill (.5,0) circle (1.5pt);
      \fill (.5,.5) circle (1.5pt);      
      \fill (1,.5) circle (1.5pt);
      \fill (1,0) circle (1.5pt);
      \fill (0,-.5) circle (1.5pt);
      \fill (.5,-.5) circle (1.5pt);
      \draw[postaction={on each segment={mid arrow}}] (0,0) -- (.5,0)--(.5,.5);
      \draw (0,-.5) node[below] {\tiny  $\bar z_1$};
      \draw (.5,-.5) node[below] {\tiny  $\bar z_2$};
      \draw (1,.5) node[right] {\tiny $z_2$};
      \draw (1,0) node[right] {\tiny $z_1$};
    \end{tikzpicture}
    \Biggr)
    +
     \WW
    \Biggl(
    \begin{tikzpicture}[baseline=-.2cm]
      \draw[postaction={on each segment={mid arrow}}] (0,-.5) -- (0,0);
      \draw[postaction={on each segment={mid arrow}}] (.5,0) -- (.5,-.5);
      \draw[postaction={on each segment={mid arrow}}] (1,0) -- (.5,0);
      \draw[postaction={on each segment={mid arrow}}] (1,.5) -- (.5,.5);
      \fill (0,0) circle (1.5pt);
      \fill (.5,0) circle (1.5pt);
      \fill (.5,.5) circle (1.5pt);      
      \fill (1,.5) circle (1.5pt);
      \fill (1,0) circle (1.5pt);
      \fill (0,-.5) circle (1.5pt);
      \fill (.5,-.5) circle (1.5pt);
      \draw[postaction={on each segment={mid arrow}}] (.5,.5)-- (.5,0)--(0,0) ;
      \draw (0,-.5) node[below] {\tiny  $\bar z_1$};
      \draw (.5,-.5) node[below] {\tiny  $\bar z_2$};
      \draw (1,.5) node[right] {\tiny $z_2$};
      \draw (1,0) node[right] {\tiny $z_1$};
    \end{tikzpicture}
    \Biggr)\\
    &= -\frac{t[q^2]\{q^{1/2}z_2\}}{\{q^{1/2}\}}+\frac{t[q\bar z_1 \bar z_2]\{q^{1/2}z_1\}}{\{q^{1/2}\}} = -\frac{t[q z_1z_2]\{q^{3/2}\bar z_1\}}{\{q^{1/2}\}}.
  \end{align}
  and
  \begin{align}
    \ZZ_n^-(z_1,z_2) &= \WW
    \Biggl(
    \begin{tikzpicture}[baseline=-.2cm]
      \draw[postaction={on each segment={mid arrow}}] (0,0) -- (0,-.5);
      \draw[postaction={on each segment={mid arrow}}] (.5,-.5) -- (.5,0);
      \draw[postaction={on each segment={mid arrow}}] (1,0) -- (.5,0);
      \draw[postaction={on each segment={mid arrow}}] (1,.5) -- (.5,.5);
      \fill (0,0) circle (1.5pt);
      \fill (.5,0) circle (1.5pt);
      \fill (.5,.5) circle (1.5pt);      
      \fill (1,.5) circle (1.5pt);
      \fill (1,0) circle (1.5pt);
      \fill (0,-.5) circle (1.5pt);
      \fill (.5,-.5) circle (1.5pt);
      \draw (0,0) -- (.5,0)--(.5,.5);
      \draw (0,-.5) node[below] {\tiny  $\bar z_1$};
      \draw (.5,-.5) node[below] {\tiny  $\bar z_2$};
      \draw (1,.5) node[right] {\tiny $z_2$};
      \draw (1,0) node[right] {\tiny $z_1$};
    \end{tikzpicture}
    \Biggr)
    =
    \WW
    \Biggl(
    \begin{tikzpicture}[baseline=-.2cm]
      \draw[postaction={on each segment={mid arrow}}] (0,0) -- (0,-.5);
      \draw[postaction={on each segment={mid arrow}}] (.5,-.5) -- (.5,0);
      \draw[postaction={on each segment={mid arrow}}] (1,0) -- (.5,0);
      \draw[postaction={on each segment={mid arrow}}] (1,.5) -- (.5,.5);
      \fill (0,0) circle (1.5pt);
      \fill (.5,0) circle (1.5pt);
      \fill (.5,.5) circle (1.5pt);      
      \fill (1,.5) circle (1.5pt);
      \fill (1,0) circle (1.5pt);
      \fill (0,-.5) circle (1.5pt);
      \fill (.5,-.5) circle (1.5pt);
      \draw[postaction={on each segment={mid arrow}}] (.5,0)--(.5,.5);
      \draw[postaction={on each segment={mid arrow}}] (.5,0)--(0,0);
      \draw (0,-.5) node[below] {\tiny  $\bar z_1$};
      \draw (.5,-.5) node[below] {\tiny  $\bar z_2$};
      \draw (1,.5) node[right] {\tiny $z_2$};
      \draw (1,0) node[right] {\tiny $z_1$};
    \end{tikzpicture}
    \Biggr)=\frac{t[q z_1z_2]\{q^{1/2}z_2\}}{\{q^{1/2}\}}.
  \end{align}
\end{example}
As in \cref{sec:GeneralisedSumComps}, we refer to the point $z_1=\dots=z_{2n}=1$ as the homogeneous limit. In this limit, the partition functions $\ZZ_n^\pm$ yield the TSASM generating function, up to a factor:
\begin{lemma}
\label{lem:TSASMGFZ}
For each $N\geqslant 0$ and $\tau = -\{q\}$, we have
  \begin{equation}
  A_{\mathrm{TS}}(2N+1;t,\tau) =[q]^{-n(2n-1)}\times
  \begin{cases}
    \ZZ^-_n(1,\dots,1), & \text{even }N,\\
    \ZZ^+_n(1,\dots,1), &\text{odd }N.
  \end{cases}
  \end{equation}
  where $n$ is defined in \eqref{eqn:Defnnbar}.
\end{lemma}
\begin{proof}
  For each $C \in \mathrm{6V}^\pm_n$, let $\tilde\mu(C)$ be the number of corner vertices with the local configuration 
  \begin{equation}
    \begin{tikzpicture}[baseline=-.2cm]
    \fill (0,0) circle (1.5pt);
    \draw[postaction={on each segment={mid arrow}}] (0,-.5) -- (0,0) -- (.5,0);
    \end{tikzpicture}
    \quad 
    \mathrm{or}
    \quad
    \begin{tikzpicture}[baseline=-.2cm]
    \fill (0,0) circle (1.5pt);
    \draw[postaction={on each segment={mid arrow}}] (.5,0) -- (0,0) -- (0,-.5);
    \end{tikzpicture},
  \end{equation}
  and $\tilde \nu(C)$ the number of bulk vertices with the local configurations
  \begin{equation}
    \begin{tikzpicture}[baseline=0cm]
       \drawvertex{5}{0}{0}
    \end{tikzpicture}
       \quad 
       \mathrm{or}
       \quad
    \begin{tikzpicture}[baseline=0cm]
       \drawvertex{6}{0}{0}
    \end{tikzpicture}
    .
  \end{equation}
  The weights \eqref{eqn:BulkWeights} and \eqref{eqn:CornerWeights} straightforwardly lead to the following expression for homogeneous limit of the partition function (note that there are $2n\times (2n-1)/2 = n(2n-1)$ bulk vertices):
  \begin{equation}
    \ZZ_n^\pm(1,\dots,1) = [q]^{n(2n-1)}\sum_{C\in \mathrm{6V}^\pm_n} t^{\tilde \mu(C)}\tau^{\tilde \nu(C)}.
  \end{equation}
  Here, $\tau = -[q^2]/[q] = -\{q\}$. Using the bijection of \cref{sec:6VTriangular}, to each $C\in \mathrm{6V}^-_n$ or $C\in \mathrm{6V}^+_n$ corresponds a unique $A\in \mathrm{TSASM}(2N+1)$ for even $N$ or odd $N$, respectively. By \eqref{eqn:6VTSASM}, we have
  \begin{equation}
    \tilde \mu(C) = \mu(A),\quad \tilde \nu(C) = \nu(A).
  \end{equation}
  The proposition follows from \eqref{eqn:TSASMGF}.
\end{proof}

\subsubsection{R- and K-matrices}

In the following, we choose to analyse the six-vertex model algebraically with the help of $R$- and $K$-matrices. Its $R$-matrix is a linear operator $\mathbb R(z) \in \mathrm{End}(\mathbb C^2 \otimes \mathbb C^2)$, depending on an indeterminate $z\in \mathbb C^\times$, whose matrix representation with respect to the standard basis of $\mathbb C^2\otimes \mathbb C^2$ is
\begin{equation}
  \RR(z)
  =
  \begin{pmatrix}
    [q\bar z] & 0 & 0 & 0\\
    0 & [qz] & -[q^2] & 0\\
    0 & -[q^2] & [qz] & 0\\
    0 & 0 & 0 & [q\bar z]
  \end{pmatrix}.
\end{equation}
We also work with a slightly modified version of the $R$-matrix. Let $P$ denote the permutation operator on $\mathbb C^2 \otimes \mathbb C^2$. It is the linear operator acting on the standard basis vectors as $P\ket{\alpha\beta} = \ket{\beta\alpha}$ for all $\alpha,\beta \in \{\uparrow,\downarrow\}$. We use it to define the $\check R$-matrix\footnote{Our definition of $\check{\RR}(z)$ differs from the usual conventions found in the literature, but greatly simplifies our presentation. }
\begin{equation}
  \check{\mathbb R}(z) = -P\mathbb R(-{\bar q}z).
\end{equation}
With this notation, the Yang-Baxter equation for the $R$-matrix takes the form
\begin{equation}
  \label{eqn:YBERR}
  \check{\mathbb R}_{1,2}(z\bar w)\mathbb R_{1,3}(z)\mathbb R_{2,3}(w) = \mathbb R_{1,3}(w)\mathbb R_{2,3}(z)\check{\mathbb R}_{1,2}(z\bar w)
\end{equation}
on $\mathbb C^2 \otimes \mathbb C^2\otimes \mathbb C^2$ for all $z,w\in \mathbb C^\times$.

The $K$-matrix is a linear operator $\mathbb K(z) \in \mathrm{End}(\mathbb C^2)$, depending on $z\in \mathbb C^\times$, whose matrix representation with respect to the standard basis is
\begin{equation}
  \KK(z) = 
  \begin{pmatrix}
    t & \{q^{1/2}z\}/\{q^{1/2}\}\\
    \{q^{1/2}z\}/\{q^{1/2}\} & t
  \end{pmatrix}.
\end{equation}
It obeys the boundary Yang-Baxter equation
\begin{equation}
  \label{eqn:bYBERR}
\check{\RR}_{1,2}(z{\bar w})\mathbb K_{1}(z)\mathbb R_{1,2}(z w)\mathbb K_{2}(w)=\mathbb K_{1}(w)\mathbb R_{1,2}(z w)\mathbb K_{2}(z)\check{\RR}_{1,2}(z{\bar w})
\end{equation}
on $\mathbb C^2 \otimes \mathbb C^2$, for all $z,w\in \mathbb C^\times$. 

To investigate the partition function \eqref{eqn:PartitionFunction6V}, we use the operator 
\begin{equation}
  \label{eqn:FactorisationM}
  \mathbb M_{1,\dots,2n}(z_1,\dots,z_{2n}) = \prod_{j=1}^{2n} {\mathbb M}^{(j)}_{j,\dots,2n}(z_1,\dots,z_{2n})
\end{equation}
on $(\mathbb C^2)^{\otimes 2n}$, depending on the model's inhomogeneity parameters. The factors of the product of the right-hand side are operators on $(\mathbb C^2)^{\otimes 2n}$ defined as
\begin{equation}
{\mathbb M}^{(j)}_{j,\dots,2n}(z_1,\dots,z_{2n})=\mathbb K_j(z_j)\prod_{k=j+1}^{2n} \mathbb R_{j,k}(z_jz_k).
\end{equation}
In the next lemma, we show that the matrix elements of this operator with respect to the canonical basis allow us to compute the partition function:
\begin{lemma}
  For each $n\geqslant 1$ and $\bm \alpha =\alpha_1\cdots \alpha_{2n}$ with $\alpha_i \in \{\uparrow,\downarrow\}$, we have
  \begin{equation}
    \ZZ^{\bm \alpha}_n(z_1,\dots,z_{2n}) = \langle \bm \alpha |\mathbb M_{1,\dots,2n}(z_1,\dots,z_{2n})\ket{\downarrow\cdots\downarrow}.
  \end{equation}
\end{lemma}
\begin{proof}
First, we note that
\begin{equation}
  \mathbb R(z_j{\bar z}_k)\ket{\gamma\delta} =\sum_{\alpha,\beta \in \{\uparrow,\downarrow\}} \WW
  \Biggl(
  \begin{tikzpicture}[scale=.5,baseline=-.075cm]
    \fill (0,0) circle (3pt);
    \draw (-1,0) -- (1,0);
    \draw (0,-1) -- (0,1);
     
    \fill[fill=white] (-.5,0) circle (7pt);
    \fill[fill=white] (.5,0) circle (7pt);
    \fill[fill=white] (0,-.5) circle (7pt);
    \fill[fill=white] (0,.5) circle (7pt);

    \draw (-.5,0) node {\tiny $\alpha$};    
    \draw (.5,0) node {\tiny $\gamma$};    
    \draw (0,-.5) node {\tiny $\beta$};    
    \draw (0,.5) node {\tiny $\delta$}; 
        
    \draw (.1,-1.35) node {\tiny $\bar z_k$};
    \draw (1.35,0) node {\tiny $z_j$};
  \end{tikzpicture}
  \hspace{-.2cm}
  \Biggr)\ket{\alpha\beta},
  \label{eqn:RRWeights}
\end{equation}
for all integers $1\leqslant j<k\leqslant 2n$ and $\gamma,\delta \in \{\uparrow,\downarrow\}$. Here and in the following, we identify the horizontal orientations $\rightarrow$ and $\leftarrow$ with $\uparrow$ and $\downarrow$, respectively. Similarly, 
  \begin{equation}
  \mathbb K(z_j)\ket{\beta}
  =
     \sum_{\alpha\in \{\uparrow,\downarrow\}} \WW  \Biggl(
  \begin{tikzpicture}[scale=.5,baseline=-.4cm]
    \fill (0,0) circle (3pt);
    \draw (0,-1) -- (0,0) -- (1,0);
    \fill[fill=white] (.5,0) circle (6pt);
    \fill[fill=white] (0,-.5) circle (6pt);
    \draw (0,-.5) node {\tiny $\alpha$};
    \draw (.5,0) node {\tiny $\beta$};
    \draw (1.3,0) node {\tiny $z_j$};
    \draw (.1,-1.3) node {\tiny $\bar z_j$};
  \end{tikzpicture}
  \hspace{-.1cm}
  \Biggr)\ket{\alpha},
  \label{eqn:KKWeights}
\end{equation}
for each $j=1,\dots,2n$ and $\beta \in \{\uparrow,\downarrow\}$.

 Second, let $1\leqslant j \leqslant 2n$ be an integer and $\delta_{j+1},\dots,\delta_{2n}\in \{\uparrow,\downarrow\}$. Using \eqref{eqn:RRWeights} and \eqref{eqn:KKWeights}, one readily obtains the identity
   \begin{multline}
 \mathbb M^{(j)}_{j,\dots,2n}(z_1,\dots,z_{2n}) \Bigl(\II^{\otimes (j-1)}\otimes \ket{\downarrow\delta_{j+1}\cdots \delta_{2n}}\Bigr) 
   \\
    = \II^{\otimes (j-1)}\otimes \sum_{\beta_j,\dots,\beta_{2n} \in \{\uparrow,\downarrow\}}
    \WW\Biggl(
    \begin{tikzpicture}[baseline=-.1cm]
      \draw (0,-.75) -- (0,0) -- (1.5,0);
      \draw[dotted] (1.5,0) -- (2.25,0);
      \draw (2.25,0) -- (3,0);
      \foreach \x in {.75,3}
      {
          \draw (\x,-.75) -- (\x,.75);
          \fill (\x,0) circle (1.5pt);
      }    
      \fill (0,0) circle (1.5pt);
      \foreach \x in {0,.75,3}
      {
         \fill[fill=white] (\x,-.375) circle (6pt);
         \fill[fill=white] (\x,.375) circle (6pt);
      }  
      \draw[postaction={on each segment={mid arrow}}] (3.75,0) --  (3,0); 
      \draw (0,-.375) node {\tiny $\beta_j$};
      \draw (.75,-.375) node {\tiny $\beta_{j+1}$};
      \draw (1.925,-.375) node {$\cdots$};
      \draw (3,-.375) node {\tiny $\beta_{2n}$};
      \draw (.75,.375) node {\tiny $\delta_{j+1}$};
      \draw (3,.375) node {\tiny $\delta_{2n}$};     
      \draw (1.925,.375) node {$\cdots$};
      \draw (0,-1) node {\tiny $\bar z_j$};     
      \draw (.75,-1) node {\tiny $\bar z_{j+1}$};     
      \draw (3,-1) node {\tiny $\bar z_{2n}$};     
      \draw (4.15,0) node {\tiny $z_j$};
      \draw (1.925,-1) node {$\cdots$};
      \end{tikzpicture}
    \Biggr)
\ket{\beta_{j}\beta_{j+1}\cdots \beta_{2n}}.
  \end{multline}
  (Recall that $\II$ denotes the identity operator on $\mathbb C^2$, i.e. the $2\times 2$ identity matrix.)
  A successive application of this identity allows us to obtain
  \begin{equation}
    \mathbb M_{1,\dots,2n}(z_1,\dots,z_{2n})\ket{\downarrow\cdots\downarrow} = \sum_{\beta_1,\dots,\beta_{2n}\in \{\uparrow,\downarrow\}} \ZZ_n^{\bm \beta}(z_1,\dots,z_n)|\bm \beta\rangle,
  \end{equation}
  where we used the abbreviation $\bm \beta = \beta_1\cdots \beta_{2n}$. The lemma follows from computing the dual pairing of $\bra{\bm \alpha}$ with both sides of this equality.
\end{proof}

\subsection{An overlap}
\label{sec:Overlap}

At present, we do not have an explicit formula for the partition function $\ZZ^{\bm \alpha}_n(z_1,\dots,z_{2n})$ with arbitrary $\bm \alpha$, even in the cases of interest $\bm \alpha = \bm \alpha^\pm$. However, by \cref{lem:TSASMGFZ}, computing the TSASM generating function only requires the knowledge of the homogeneous limit $\ZZ^{\pm}_n(1,\dots,1)$. To obtain this homogeneous limit, we study a linear combination of half-specialised partition functions $\ZZ^{\bm \alpha}_n(w_1,\bar w_1,\dots,w_{n},\bar w_n)$ with selected $\bm \alpha$. This linear combination exhibits a particularly rich structure that enables us to make further progress. For this purpose, we define the covector
\begin{equation}
  \bra{\nu(w)} = [{\bar q}bw]\bra{\uparrow\downarrow}+[qb{\bar w}]\bra{\downarrow\uparrow} \in (\mathbb C^2\otimes \mathbb C^2)^\ast,
  \label{eqn:DefNu}
\end{equation}
where $b\in \mathbb C^\times$ is an additional parameter. Moreover, for each $n\geqslant 1$, we use the abbreviation
\begin{equation}
  \bra{\nu_n(w_1,\dots,w_n)} = \bigotimes_{i=1}^n \bra{\nu(w_i)}.
\end{equation}
Using this covector, we define for each $n\geqslant 1$ the overlap
\begin{equation}
  \label{eqn:DefZZ}
  \ZZ_n(w_1,\dots,w_{n}) = \langle \nu_{n}(w_1,\dots,w_{n})|\mathbb M_{1,\dots,2n}(w_1,\bar{w}_1,\dots,w_n,\bar{w}_n)\ket{\downarrow\cdots\downarrow}.
\end{equation}
It follows from \eqref{eqn:DefNu} that this overlap is a linear combination of the half-specialised partition functions $\ZZ^{\bm \alpha}_n(w_1,\bar w_1,\dots,w_{n},\bar w_n)$, where $\bm \alpha$ runs over $2^n$ distinct spin configurations with $\alpha_{2i-1}\alpha_{2i}=\,\uparrow\downarrow$ or $\downarrow\uparrow$, for each $i=1,\dots,n$. Furthermore, we define $\ZZ_0=1$.

\begin{example}
 \label{ex:ZZ1}
 For $n=1$, using the explicit expressions of \cref{ex:Zpm1}, we obtain \begin{multline}
   \ZZ_1(w_1)  = [{\bar q}bw_1]\ZZ_1^+(w_1,\bar w_1)+[q b{\bar w_1}]\ZZ_1^-(w_1,\bar w_1) = t[q^{1/2}]\{bq^{-1/2}\}[q^2\bar{w}_1^2].
 \end{multline}
 For $n=2$, we have
 \begin{multline}
   \ZZ_2(w_1,w_2) = [{\bar q}bw_1][{\bar q}bw_2]\ZZ_2^{+}(w_1,\bar w_1,w_2,\bar w_2)+[qb{\bar w}_1][{\bar q}bw_2]\ZZ_2^{\downarrow\uparrow\uparrow\downarrow}(w_1,\bar w_1,w_2,\bar w_2)\\ + [{\bar q}bw_1][qb{\bar w}_2]\ZZ_2^{\uparrow\downarrow\downarrow\uparrow}(w_1,\bar w_1,w_2,\bar w_2)+[qb{\bar w}_1][qb{\bar w}_2]\ZZ_2^{-}(w_1,\bar w_1,w_2,\bar w_2) , \end{multline}
 whose explicit expression is, however, rather complicated.
\end{example}
We also note that $\ZZ_n$ can be identified with Kuperberg's partition function on a triangular graph with a UOS-boundary \cite{kuperberg:02}. However, Kuperberg assigns weights to the local configurations at the corner vertices that are independent of the inhomogeneity parameters. By contrast, some of our weights depend on these parameters and, thus, vary from one corner vertex to another.

We now examine several properties of $\ZZ_n$, closely following the structure of \cref{sec:PropGenSum}. 
\subsubsection{Symmetry and degree}
We use the following lemma to analyse the symmetries of $\ZZ_n$:
\begin{lemma}
  \label{lem:RM}
  For each $i=1,\dots, 2n-1$, we have 
  \begin{multline}
    \label{eqn:RM}
    {\check {\mathbb R}}_{i,i+1}(z_i\bar{z}_{i+1})\mathbb M_{1,\dots,2n}(\dots,z_i,z_{i+1},\dots)\\=  \mathbb M_{1,\dots,2n}(\dots,z_{i+1},z_{i},\dots){\check {\mathbb R}}_{i,i+1}(z_i\bar{z}_{i+1}).
  \end{multline}
\end{lemma}
\begin{proof}
Consider the left-hand side of \eqref{eqn:RM} and recall the factorisation \eqref{eqn:FactorisationM}. We commute $\check {\mathbb R}_{i,i+1}(z_i\bar{z}_{i+1})$ through the factors $\mathbb M^{(j)}_{j,\dots,2n}(z_1,\dots,z_{2n})$ for each $j=1,\dots,n$. We distinguish the following two cases.
\medskip

\noindent \textit{Case (i):} $i<j-1$ or $i>j$. In this case, we claim that
\begin{multline}
\label{eqn:RMM1}
{\check {\mathbb R}}_{i,i+1}(z_i\bar{z}_{i+1})\mathbb M^{(j)}_{j,\dots,2n}(\dots,z_i,z_{i+1},\dots)\\
=\mathbb M^{(j)}_{j,\dots,2n}(\dots,z_{i+1},z_i,\dots){\check {\mathbb R}}_{i,i+1}(z_i\bar{z}_{i+1}).
\end{multline}
For $i<j-1$, this relation holds trivially because both the $\check R$-matrix and the operator $\mathbb M^{(j)}_{j,\dots,2n}(\dots,z_i,z_{i+1},\dots)$ act non-trivially on disjoint tensor factors of $(\mathbb C^2)^{\otimes 2n}$ and, moreover, $\mathbb M^{(j)}_{j,\dots,2n}(\dots,z_i,z_{i+1},\dots)$ is independent of $z_i,z_{i+1}$. 

For $i>j$, we observe that ${\check {\mathbb R}}_{i,i+1}(z_i\bar{z}_{i+1})$ trivially commutes with the $K$-matrix and the first $i-j-1$ factors of the product in \eqref{eqn:FactorisationM}. Hence,
\begin{multline}
  \label{eqn:RMIntermediate2}
{\check {\mathbb R}}_{i,i+1}(z_i\bar{z}_{i+1})\mathbb M^{(j)}_{j,\dots,2n}(\dots,z_i,z_{i+1},\dots)=\mathbb K_j(z_j) \prod_{k=j+1}^{i-1}\mathbb R_{j,k}(z_jz_k)\\
\times \left({\check {\mathbb R}}_{i,i+1}(z_i\bar{z}_{i+1})\mathbb R_{j,i}(z_jz_i)\mathbb R_{j,i+1}(z_jz_{i+1})\right)\prod_{k=i+2}^{2n}\mathbb R_{j,k}(z_jz_k).
\end{multline}
We use the Yang-Baxter equation \eqref{eqn:YBERR} to rewrite
\begin{equation}
  {\check {\mathbb R}}_{i,i+1}(z_i\bar{z}_{i+1})\mathbb R_{j,i}(z_jz_i)\mathbb R_{j,i+1}(z_jz_{i+1})=\mathbb R_{j,i}(z_jz_{i+1})\mathbb R_{j,i+1}(z_jz_i){\check{\mathbb R}}_{i,i+1}(z_i\bar{z}_{i+1}).
\end{equation}
Next, we pull ${\check{\mathbb R}}_{i,i+1}(z_i\bar{z}_{i+1})$ through the rightmost product in \eqref{eqn:RMIntermediate2}, as the two operators act on different tensor factors of $(\mathbb C^2)^{\otimes 2n}$. Recomposing the different products leads to \eqref{eqn:RMM1}.

\medskip

\noindent \textit{Case (ii):} $i=j,j+1$. We have
\begin{multline}
  {\check {\mathbb R}}_{i,i+1}(z_i\bar{z}_{i+1})\mathbb M^{(i)}_{i,\dots,2n}(\dots,z_i,z_{i+1},\dots)\mathbb M^{(i+1)}_{i+1,\dots,2n}(\dots,z_i,z_{i+1},\dots)\\
  = {\check {\mathbb R}}_{i,i+1}(z_i\bar{z}_{i+1})\mathbb K_i(z_i){\mathbb R}_{i,i+1}(z_iz_{i+1})\mathbb K_{i+1}(z_{i+1})\prod_{k=i+2}^{2n} \mathbb R_{i,k}(z_iz_k)\mathbb R_{i+1,k}(z_{i+1}z_k).\nonumber
\end{multline}
We use the boundary Yang-Baxter equation \eqref{eqn:bYBERR} to write
\begin{multline}
  {\check {\mathbb R}}_{i,i+1}(z_i\bar{z}_{i+1})\mathbb K_i(z_i){\mathbb R}_{i,i+1}(z_iz_{i+1})\mathbb K_{i+1}(z_{i+1})\\
  = \mathbb K_i(z_{i+1}){\mathbb R}_{i,i+1}(z_iz_{i+1})\mathbb K_{i+1}(z_{i}){\check {\mathbb R}}_{i,i+1}(z_i\bar{z}_{i+1}).
\end{multline}
Moreover, by \eqref{eqn:YBERR}, we have
\begin{multline}
  {\check {\mathbb R}}_{i,i+1}(z_i\bar{z}_{i+1})\left( \prod_{k=i+2}^{2n} \mathbb R_{i,k}(z_iz_k)\mathbb R_{i+1,k}(z_{i+1}z_k)\right)\\
  = \left(\prod_{k=i+2}^{2n} \mathbb R_{i,k}(z_{i+1}z_k)\mathbb R_{i+1,k}(z_iz_k)\right){\check {\mathbb R}}_{i,i+1}(z_i\bar{z}_{i+1}).
\end{multline}
Combining these two observations, we obtain
\begin{multline}
  \label{eqn:RMM2}
  {\check {\mathbb R}}_{i,i+1}(z_i\bar{z}_{i+1})\mathbb M^{(i)}_{i,\dots,2n}(\dots,z_i,z_{i+1},\dots)\mathbb M^{(i+1)}_{i+1,\dots,2n}(\dots,z_i,z_{i+1},\dots)\\
  =  \mathbb M^{(i)}_{i,\dots,2n}(\dots,z_{i+1},z_i,\dots)\mathbb M^{(i+1)}_{i+1,\dots,2n}(\dots,z_{i+1},z_i,\dots){\check {\mathbb R}}_{i,i+1}(z_i\bar{z}_{i+1}).
\end{multline}
The lemma follows from combining \eqref{eqn:RMM1} and \eqref{eqn:RMM2}.
\end{proof}

\begin{proposition}
  \label{prop:ZZsym}
  For each $n\geqslant 2$, $\ZZ_{n}$ is a symmetric function of $w_1,\dots,w_n$.
\end{proposition}
\begin{proof}
 The proof relies on the identity
 \begin{multline}
   \left(\bra{\nu(w)}\otimes \bra{\nu(z)}\right)\check \RR_{2,3}(zw) \check \RR_{1,2}(z\bar w)\check \RR_{3,4}(\bar z w)  \check \RR_{2,3}(\bar z \bar w) \\
   =r(zw)r(z\bar w)\bra{\nu(z)}\otimes \bra{\nu(w)},
   \label{eqn:bYBENu}
 \end{multline}
 where $r(z) = [q^2z][q^2{\bar z}]$, which can be checked by explicit calculation.
 
 Let $1\leqslant i \leqslant n-1$ be an integer. Generalising the product of $\check R$-matrices in \eqref{eqn:bYBENu}, we define the following operator on $(\mathbb C^2)^{\otimes 2n}$:
  \begin{multline}
    \mathbb O_{2i-1,\dots,2i+2} = \check{\RR}_{2i,2i+1}(w_iw_{i+1})\check{\RR}_{2i-1,2i}(w_i\bar w_{i+1})\\
    \times \check{\RR}_{2i+1,2i+2}(w_{i+1}\bar w_i)\check{\RR}_{2i,2i+1}(\bar w_i\bar w_{i+1}).
   \end{multline}
   We exploit three simple properties of this operator. First, \eqref{eqn:bYBENu} implies
   \begin{multline}
     \label{eqn:OInt1}
     r(w_i\bar w_{i+1})r(w_iw_{i+1})  \bra{\nu_n(\dots,w_i,w_{i+1},\dots)}\\=\bra{\nu_n(\dots,w_{i+1},w_i,\dots)} \mathbb O_{2i-1,\dots,2i+2}.
   \end{multline}
   It allows us to write
   \begin{multline}
     r(w_i\bar w_{i+1})r(w_iw_{i+1}) \ZZ_n(\dots,w_i,w_{i+1},\dots) \\
     = \bra{\nu_n(\dots,w_{i+1},w_i,\dots)}\mathbb O_{2i-1,\dots,2i+2}\mathbb M_{1,\dots,2n}(\dots,w_{i},\bar w_{i},w_{i+1},\bar w_{i+1},\dots)\ket{\downarrow\cdots\downarrow}.\nonumber
   \end{multline}
   Second, using \cref{lem:RM} repeatedly, we obtain the commutation relation
   \begin{multline}
     \label{eqn:OInt2}
     \mathbb O_{2i-1,\dots,2i+2}\mathbb M_{1,\dots,2n}(\dots,w_{i},\bar w_{i},w_{i+1},\bar w_{i+1},\dots)\\
     =\mathbb M_{1,\dots,2n}(\dots,w_{i+1},\bar w_{i+1},w_{i},\bar w_{i},\dots)\mathbb O_{2i-1,\dots,2i+2}.
   \end{multline} 
   It leads to 
   \begin{multline}
     r(w_i\bar w_{i+1})r(w_iw_{i+1}) \ZZ_n(\dots,w_i,w_{i+1},\dots) \\
     = \bra{\nu_n(\dots,w_{i+1},w_i,\dots)}\mathbb M_{1,\dots,2n}(\dots,w_{i+1},\bar w_{i+1},w_{i},\bar w_{i},\dots)\mathbb O_{2i-1,\dots,2i+2}\ket{\downarrow\cdots\downarrow}.\nonumber
   \end{multline}
   Third, one verifies $\mathbb O_{2i-1,\dots,2i+2}\ket{\downarrow\cdots\downarrow}=r(w_i\bar w_{i+1})r(w_iw_{i+1})\ket{\downarrow\cdots\downarrow}$. Hence, we obtain the equality
   \begin{multline}
   r(w_i\bar w_{i+1})r(w_iw_{i+1}) \ZZ_n(\dots,w_i,w_{i+1},\dots)\\=r(w_i\bar w_{i+1})r(w_iw_{i+1}) \ZZ_n(\dots,w_{i+1},w_i,\dots).
   \end{multline}
 Since $\ZZ_n(\dots,w_{i},w_{i+1},\dots)$ is a Laurent polynomial, we conclude
   \begin{equation}
     \ZZ_n(\dots,w_{i+1},w_{i},\dots)=\ZZ_n(\dots,w_{i},w_{i+1},\dots).
   \end{equation}
   This equality holds for each $i=1, \dots, n-1$, which straightforwardly implies the symmetry of $\ZZ_n$.
\end{proof}

\begin{proposition}
  \label{prop:ZZinv}
  For each $i=1,\dots,n$, we have
  \begin{equation}
    [q^2 w_i^2] {\ZZ}_n(\dots,w_i,\dots) = [q^2\bar{w}_i^2] {\ZZ}_n(\dots,\bar w_i,\dots).
  \end{equation}
\end{proposition}
\begin{proof}
  By \cref{prop:ZZsym}, it is sufficient to consider $i=1$. We use the identity
  \begin{equation}
    \bra{\nu(\bar{z})}\check \RR(z^2) = [q^2 z^2]\bra{\nu(z)},
  \end{equation} 
  which can be checked by direct calculation. It allows us to write
  \begin{align}
    [q^2 w_i^2] \ZZ_n(\dots,w_i,\dots) = \bra{\nu_n(\bar{w}_1,\dots)}\check \RR_{1,2}(w_i^2)\mathbb M_{1,\dots,2n}(w_1,\bar{w}_1,\dots)\ket{\downarrow\cdots \downarrow}.
  \end{align}
  Next, we commute $\RR_{1,2}(w_i^2)$ to the right of $\mathbb M_{1,\dots,2n}(w_1,{\bar w}_1,\dots)$, using \cref{lem:RM}, which yields
  \begin{align}
    [q^2 w_i^2] Z_n(\dots,w_i,\dots) = \bra{\nu_n(\bar{w}_1,\dots)}\mathbb M_{1,\dots,2n}(\bar{w}_1,w_1,\dots)\check\RR_{1,2}(w_i^2)\ket{\downarrow\cdots \downarrow}.
  \end{align}
  The proposition follows from the action $\check \RR_{1,2}(w_i^2)\ket{\downarrow\downarrow\dots \downarrow} = [q^2\bar{w}_i^2]\ket{\downarrow\downarrow\dots \downarrow}$.
\end{proof}

\begin{proposition}
  \label{prop:ZZop}
  For each $i=1,\dots,n$, we have $\ZZ_n(\dots,-w_i,\dots) = \mathbb Z_n(\dots,w_i,\dots)$.
\end{proposition}
\begin{proof}
  By \cref{prop:ZZsym}, it is sufficient to consider $i=1$. We note the properties \begin{equation}
    \mathbb K(-w) = \sigma^z \mathbb K(w)\sigma^z, \quad \RR(-z) = -(\sigma^z \otimes \II)\RR(z)(\sigma^z\otimes \II).
  \end{equation}
  They straightforwardly imply
  \begin{equation}
    {\mathbb M}_{1,\dots,2n}(-w_1,-\bar{w}_1,\dots) = \sigma_1^z\sigma_2^z {\mathbb M}_{1,\dots,2n}(w_1,\bar{w}_1,\dots)\sigma_1^z\sigma_2^z.
  \end{equation}
  Hence,
  \begin{equation}
    \ZZ_n(-w_1,\dots) = \bra{\nu_n(-w_1,\dots)}\sigma_1^z\sigma_2^z {\mathbb M}_{1,\dots,2n}(w_1,{\bar w}_1,\dots)\sigma_1^z\sigma_2^z\ket{\downarrow\downarrow\cdots \downarrow}.
  \end{equation}
  The proposition follows from the identity $\bra{\nu(-w)}(\sigma^z\otimes \sigma^z)= \bra{\nu(w)}$, which is immediate, and $\sigma_1^z\sigma_2^z\ket{\downarrow\downarrow\cdots \downarrow}=\ket{\downarrow\downarrow\cdots \downarrow}$.
\end{proof}

\Cref{prop:ZZsym,prop:ZZinv,prop:ZZop} justify defining
\begin{equation}
  \label{eqn:DefYY}
  \YY_n(w_1,\dots,w_n) = \frac{(-1)^{n(n+1)/2}}{[q^{1/2}]^n\prod_{i=1}^n [q^2\bar{w}_i^2]} \mathbb Z_n(w_1,\dots,w_n).
\end{equation}

\begin{example}
 \label{ex:YY1}
For $n=0$, we have $\YY_0=\ZZ_0=1$. For $n=1$, using the expression given in \cref{ex:ZZ1}, we have
\begin{equation}
\YY_1(w_1) = \frac{-1}{[q^{1/2}][q^2\bar w_1^2]}\ZZ_1(w_1) = -t\{bq^{-1/2}\}.
\end{equation}
\end{example}

\begin{proposition}
  \label{prop:YYBC}
 For each $n\geqslant 1$, $\YY_n$ is an even $BC_n$-symmetric Laurent polynomial in $w_1,\dots,w_n$. For each $i=1,\dots, n$, its degree width as a Laurent polynomial in $w_i$ is at most $8(n-1)$.
\end{proposition}
\begin{proof}
  Clearly, $\ZZ_n$ is a Laurent polynomial in $w_1,\dots,w_n$. The definition \eqref{eqn:DefYY} and \cref{prop:ZZsym,prop:ZZinv,prop:ZZop} straightforwardly imply that $\YY_n$ is an even $BC_n$-symmetric Laurent polynomial in $w_1,\dots,w_n$.
  
  To find the degree width, we note that the leading term of $\ZZ_n$ as a Laurent polynomial in $w_1$ is of degree at most $4n-1$ by construction. This bound on the degree can be improved to $4n-2$, because $\ZZ_n$ is an even function of $w_1$ by \cref{prop:ZZop}. It follows from \eqref{eqn:DefYY}, that the leading term of $\YY_n$ as a Laurent polynomial in $w_1$ is of degree at most $4(n-1)$. The bound $8(n-1)$ for the degree width of $\YY_n$ in $w_1$, and consequently in each $w_i$, follows from the $BC_n$-symmetry.
  \end{proof}

\subsubsection{Reduction}
In this section, we establish two reduction relations for $\YY_n$.

\begin{proposition}
  \label{prop:YYRed1}
  For each $i=1,\dots,n$, we have
  \begin{multline}
  \YY_n(\dots,w_i= \ii q^{1/2},\dots) = (-1)^nt\{bq^{-1/2} \}\\
  \times \Biggl(\prod_{j\neq i}^n \{q^{3/2} w_j\}^2 \{q^{3/2}\bar w_i\}^2\Biggr) \YY_{n-1}(\dots,\widehat{w_i},\dots)
  \end{multline}
  where $\widehat{\cdots}$ denotes omission.
\end{proposition}
\begin{proof}
  For $n=1$, the proposition straightforwardly follows from \cref{ex:YY1}. For $n>1$, it is sufficient to consider the case where $i=n$ by \cref{prop:YYBC}. Note that
   \begin{equation}
     \label{eqn:ZZFromPhi}
    \ZZ_n(w_1,\dots,w_{n-1},\ii q^{1/2}) = \bra{\nu_{n-1}(w_1,\dots,w_{n-1})}\phi(w_1,\bar w_1,\dots,w_{n-1},\bar w_{n-1})\rangle,
  \end{equation}
  with a vector $\ket{\phi}= \ket{\phi(z_1,\dots,z_{2n-2})} \in (\mathbb C^2)^{\otimes (2n-2)}$, given by
   \begin{equation}
     \ket{\phi} = (\II^{\otimes (2n-2)}
\otimes \bra{\nu(\ii q^{1/2})})\MM_{1,\dots,2n}(z_1,\dots,z_{2n-2},\ii q^{1/2},-\ii q^{-1/2})\ket{\underset{2n}{\underbrace{\downarrow\cdots \downarrow}}}.
   \end{equation}
  We evaluate this vector with the help of the alternative factorisation
  \begin{equation}
    \MM_{1,\dots,2n}(z_1,\dots,z_{2n}) = \prod_{j=1}^{2n} \left(\prod_{i=1}^{j-1} \RR_{i,j}(z_iz_j) \right) \KK_j(z_j).
\end{equation}
   Upon specialisation $z_{2n-1} = \ii q^{1/2},\,z_{2n}=-\ii q^{-1/2}$, this factorisation allows us to write   
\begin{multline}
     \ket{\phi} 
     =\MM_{1,\dots,2n-2}(z_1,\dots,z_{2n-2})(\II^{\otimes (2n-2)}
 \otimes \bra{\nu(\ii q^{1/2})})\prod_{i=1}^{2n-2} \RR_{i,2n-1}(\ii q^{1/2}z_i)\\
     \quad \times \KK_{2n-1}(\ii q^{1/2})\prod_{i=1}^{2n-2} \RR_{i,2n}(-\ii z_i q^{-1/2}) \mathbb R_{2n-1,2n}(1) \KK_{2n}(-\ii q^{-1/2})\ket{\underset{2n}{\underbrace{\downarrow\cdots \downarrow}}}.
   \end{multline}
   Here, $\MM_{1,\dots,2n-2}(z_1,\dots,z_{2n-2})$ should be understood as an operator on $(\mathbb C^2)^{\otimes (2n-2)}$. We use $\KK(-\ii q^{-1/2}) = t \II$ to obtain
  \begin{multline}
    \prod_{i=1}^{2n-2} \RR_{i,2n}(-\ii z_i q^{-1/2}) \mathbb R_{2n-1,2n}(1) \KK_{2n}(-\ii q^{-1/2})\ket{\underset{2n}{\underbrace{\downarrow\cdots \downarrow}}}\\  = (-1)^{n-1}t[q]\prod_{i=1}^{2n-2}\{q^{3/2}\bar z_i\}\ket{\underset{2n}{\underbrace{\downarrow\cdots \downarrow}}}.
  \end{multline}
  Combining this evaluation with $\bra{\nu(\ii q^{1/2})}\left(\II \otimes \ket{\downarrow}\right)=\ii\{bq^{-1/2}\}\bra{\uparrow}$, we find
   \begin{multline}
       \ket{\phi}= (-1)^{n-1}\ii\{bq^{-1/2}\}t[q]\prod_{i=1}^{2n-2}\{q^{3/2}\bar z_i\}\MM_{1,\dots,2n-2}(z_1,\dots,z_{2n-2})\\
	\times (\II^{\otimes (2n-2)}
\otimes \bra{\uparrow})\prod_{i=1}^{2n-2} \RR_{i,2n-1}(\ii z_i q^{1/2}) \KK_{2n-1}(\ii q^{1/2})\ket{\underset{2n-1}{\underbrace{\downarrow\cdots \downarrow}}}.
   \end{multline}
  Next, we repeatedly use the identity
  \begin{equation}
    (\II\otimes \bra{\uparrow})\RR(z)(\ket{\downarrow}\otimes \II) = [q z](\II\otimes \bra{\uparrow}) (\ket{\downarrow}\otimes \II),
  \end{equation}
  and $\bra{\uparrow}\mathbb K(\ii q^{1/2})\ket{\downarrow}=\ii[q]/\{q^{1/2}\}$, to obtain
  \begin{multline}
    (\II^{\otimes (2n-2)} \otimes \bra{\uparrow})\prod_{i=1}^{2n-2} \RR_{i,2n-1}(\ii z_i q^{1/2}) \KK_{2n-1}(\ii q^{1/2})\ket{\underset{2n-1}{\underbrace{\downarrow\cdots \downarrow}}}\\
    =(-1)^{n-1}\ii[q]/\{q^{1/2}\}\prod_{i=1}^{2n-2}\{q^{3/2}z_i\}\ket{\underset{2n-2}{\underbrace{\downarrow\cdots\downarrow}}}.
  \end{multline}
  Hence, we obtain
  \begin{multline}
    |\phi\rangle = -\{bq^{-1/2}\}[q]^2t/\{q^{1/2}\}\prod_{i=1}^{2n-2}\{q^{3/2}z_i\}\{q^{3/2}\bar z_i\}\\
    \times \MM_{1,\dots,2n-2}(z_1,\dots,z_{2n-2})\ket{\underset{2n-2}{\underbrace{\downarrow\cdots\downarrow}}}.
  \end{multline}
  Upon half-specialisation $z_{2i-1} = w_i,\,z_{2i} = \bar w_i,\,i=1,\dots,n-1$ and substitution into \eqref{eqn:ZZFromPhi}, this expression straightforwardly leads to  \begin{multline}
    \mathbb Z_n(w_1,\dots,w_{n-1},\ii q^{1/2}) = -\{bq^{-1/2}\}[q]^2t/\{q^{1/2}\}\prod_{i=1}^{2n-2}\{q^{3/2}w_i\}^2\{q^{3/2}\bar w_i\}^2\\
    \times \mathbb Z_{n-1}(w_1,\dots,w_{n-1}).
  \end{multline}
  The proposition follows from this relation and \eqref{eqn:DefYY}.
\end{proof}

For the next reduction relation, it will be useful to introduce the abbreviation
\begin{equation}
  F(w) = - [q^2]^2 [b\bar w][b\bar q w] \{q^{1/2} w\}\{q^{3/2}\bar w\}\det \KK(\bar w).
  \label{eqn:DefF}
\end{equation}

\begin{proposition}
  \label{prop:YYRed2}
  For $n\geqslant 2$ and all $1 \leqslant i < j\leqslant n$, we have
  \begin{multline}
    \label{eqn:ReductionY}
    \YY_n(\dots,w_i,\dots,w_j = \bar q w_i,\dots) = F(w_i) \\ \times\Biggl(\prod_{k\neq i,j}^n[q^2 \bar w_i \bar w_k][q^2\bar w_i w_k][qw_iw_k][q w_i\bar w_k]\Biggr)^2
   \YY_{n-2}(\dots,\widehat{w_i},\dots,\widehat{w_j},\dots).
  \end{multline}
\end{proposition}
\begin{proof}
  By \cref{prop:YYBC}, it is sufficient to consider $i=n-1$ and $j=n$. Moreover, we have    \begin{equation}
  \YY_n(\dots,w_{n-1},w_n = \bar q w_{n-1}) = \YY_n(\dots,w_{n-1},w_n = -q \bar w_{n-1}).
  \end{equation}
  It will be convenient to work with the choice of arguments on the right-hand side. To simplify the presentation, we write $w_{n-1} = w$. The reasoning has two parts. First, we consider in detail the case where $n=2$, as it presents all necessary ingredients. Second, we use it to address the cases where $n>2$.
    
  For $n=2$, we need to evaluate the overlap
  \begin{equation}
    \mathbb Z_2(w,-q\bar w) = \left(\bra{\nu(w)}\otimes \bra{\nu(-q\bar w)}\right)\MM_{1,2,3,4}(w,\bar w,-q\bar w,-\bar q w)\ket{\downarrow\downarrow\downarrow\downarrow}.
  \end{equation}
  Explicitly, we have
  \begin{multline}
    \MM_{1,2,3,4}(w,\bar w,-q\bar w,-\bar q w) = \KK_1(w)\RR_{1,2}(1)\RR_{1,3}(-q)\RR_{1,4}(-\bar q w^2)\KK_2(\bar w)\\
    \times \RR_{2,3}(-q\bar{w}^2)\RR_{2,4}(-\bar q)\KK_3(-q\bar w)\RR_{3,4}(1)\KK_4(-\bar{q}w).
    \label{eqn:MM1234Intermediate}
  \end{multline}
  We use the simplification $\RR(-\bar q)=-[q^2]P$ and the boundary Yang-Baxter equation \eqref{eqn:bYBERR} to rewrite the product of $R$- and $K$-matrices on the second line of this equality as follows:
  \begin{multline}
    \RR_{2,3}(-q\bar{w}^2)\RR_{2,4}(-\bar q)\KK_3(-q\bar w)\RR_{3,4}(1)\KK_4(-\bar{q}w)\\
    = -[q^2]\KK_2(-\bar q w) \RR_{2,3}(1)\KK_3(-q\bar w)\RR_{2,3}(-q\bar{w}^2)P_{2,4}.
  \end{multline}
  Upon substitution into \eqref{eqn:MM1234Intermediate}, we simplify this expression
 with the help of the identity $\KK(\bar w)\KK(-\bar q w) = (\det \KK(\bar w)) \II$, which follows from a straightforward calculation, and the Yang-Baxter equation \eqref{eqn:YBERR}:
  \begin{multline}
    \MM_{1,2,3,4}(w,\bar w,-q\bar w,-\bar q w) = -[q^2](\det \KK(\bar w))\KK_1(w)\RR_{2,3}(1)\RR_{1,3}(-q)\\
    \times \KK_3(-q\bar w)\RR_{1,2}(1)\RR_{2,3}(-q\bar{w}^2)\RR_{1,4}(-\bar q w^2)P_{2,4}.
  \end{multline}
  We now act with this expression on $\ket{\downarrow\downarrow\downarrow\downarrow}$. The action of the last four factors readily follows from $\RR(z)\ket{\downarrow\downarrow} = [q\bar z]\ket{\downarrow\downarrow}$:
  \begin{multline}
   \MM_{1,2,3,4}(w,\bar w,-q\bar w,-\bar q w)\ket{\downarrow\downarrow\downarrow\downarrow} = 
   -[q][q^2][w^2][q^2\bar w^2](\det \KK(\bar w))\\
    \times \KK_1(w)\RR_{2,3}(1)\RR_{1,3}(-q)\KK_3(-q\bar w)\ket{\downarrow\downarrow\downarrow\downarrow}.
  \end{multline}
  We observe that the product of $R$- and $K$-matrices on the second line acts like the identity on the last tensor factor of $(\mathbb C^2)^{\otimes 4}$. Combining this observation with \eqref{eqn:DefNu}, we obtain
    \begin{multline}
     \mathbb Z_2(w,-q\bar w)=[q][q^2][b\bar w][w^2][q^2\bar w^2](\det \KK(\bar w))\\
     \times \left(\bra{\nu(w)}\otimes \bra{\uparrow}\right)\KK_1(w)\RR_{2,3}(1)\RR_{1,3}(-q)\KK_3(-q\bar w)\ket{\downarrow\downarrow\downarrow}.
     \label{eqn:MM1234Intermediate1}
  \end{multline}
  By \eqref{eqn:DefNu}, the non-zero contributions to the remaining overlap on the second line result from $\KK_1(w)\RR_{2,3}(1)\RR_{1,3}(-q)\KK_3(-q\bar w)$ flipping two spins $\downarrow$ to $\uparrow$ when acting on $\ket{\downarrow\downarrow\downarrow}$. The flips are realised thanks to the non-zero upper-right entries of the $K$-matrices. It follows that we may replace $\KK_1(w)$ in \eqref{eqn:MM1234Intermediate1} by $\{q^{1/2}w\}/\{q^{1/2}\}\sigma_1^+$ and, likewise, $\KK_3(-q\bar w)$ by $-\{q^{3/2}\bar w\}/\{q^{1/2}\}\sigma_3^+$, where $\sigma^+ = \frac12(\sigma^x+\ii \sigma^y)$, without changing the overlap. Thanks to $\sigma^+\ket{\downarrow}=\ket{\uparrow}$, $\bra{\uparrow}\sigma^+=\bra{\downarrow}$, and $\bra{\downarrow}\sigma^+=0$, the expression simplifies to 
   \begin{multline}
     \mathbb Z_2(w,-q\bar w)=-[q][q^2]\{q^{1/2}w\}\{q^{3/2}\bar w\}[b\bar w][b\bar q w][w^2][q^2\bar w^2](\det \KK(\bar w))/\{q^{1/2}\}^2\\
     \times \bra{\downarrow\downarrow\uparrow}\RR_{2,3}(1)\RR_{1,3}(-q)\ket{\downarrow\downarrow\uparrow}.
         \label{eqn:MM1234Intermediate2}
  \end{multline}
We note that 
  \begin{equation}\label{eqn:MM1234 matrix element}
 (\II\otimes \bra{\uparrow})\RR_{2,3}(1)\RR_{1,3}(-q)\ket{\downarrow\downarrow\uparrow}=-[q][q^2]\ket{\downarrow\downarrow}.
\end{equation}
The substitution into \eqref{eqn:MM1234Intermediate2} yields the compact expression
  \begin{equation}
     \mathbb Z_2(w,-q\bar w)=-[q^{1/2}]^2[w^2][q^2\bar w^2]F(w).
         \label{eqn:MM1234Intermediate3}
  \end{equation}
  Finally, we obtain \eqref{eqn:ReductionY} for $n=2$ by combining this result with \eqref{eqn:DefYY}.
  
  We now consider the cases where $n>2$. Clearly, we have
  \begin{multline}
    \ZZ_n(w_1,\dots,w_{n-2},w,-q\bar w) = \bra{\nu_{n-2}(w_1,\dots,w_{n-2})}\phi(w_1,\bar w_1,\dots,w_{n-2},\bar w_{n-2})\rangle
    \label{eqn:ZZPhi}
  \end{multline}
  with a vector $\ket{\phi(z_1,\dots,z_{2n-4})} \in (\mathbb C^{2})^{\otimes (2n-4)}$, defined through
  \begin{multline}
     |\phi(z_1,\dots,z_{2n-4})\rangle = (
\II^{\otimes (2n-4)}
 \otimes \bra{\nu(w)} \otimes \bra{\nu(-q\bar w)})\\
     \MM_{1,\dots,2n}(z_1,\dots,z_{2n-4},w,\bar w,-q\bar w,-\bar q w)\ket{\downarrow\cdots \downarrow}.
  \end{multline} 
  We may write
  \begin{align}
    \MM_{1,\dots,2n}&(z_1,\dots,z_{2n-4},w,\bar w,-q\bar w,-\bar q w) \notag\\
    &= \quad \left(\MM_{1,\dots, 2n-4} (z_1,\dots,z_{2n-4})\otimes \II^{\otimes 4}\right)\notag\\
    &\quad \times \prod_{i=1}^{2n-4}\RR_{i,2n-3}(z_iw)\RR_{i,2n-2}(z_i\bar w)\RR_{i,2n-1}(-qz_i\bar w)\RR_{i,2n}(-\bar qz_i w) \label{eqn:RMatrixProduct}\\
    &\quad \times\left( \II^{\otimes (2n-4)}\otimes \MM_{1,2,3,4}(w,\bar w,-q\bar w,-\bar q w)\right).\notag
  \end{align}
  Using this factorisation, we evaluate $\ket{\phi(z_1,\dots,z_{2n-4})}$ by following the steps for the case $n=2$ that lead to \eqref{eqn:MM1234Intermediate3}. Observing (at various steps of the evaluation) that the action of the $R$-matrices of each factor of the product in \eqref{eqn:RMatrixProduct} is diagonal thanks to the vectors they act on, to both their left and right, and using the identity \eqref{eqn:MM1234 matrix element}, we obtain
  \begin{multline}
    \ket{\phi(z_1,\dots,z_{2n-4})} =-[q^{1/2}]^2[w^2][q^2\bar w^2]F(w)  \\
    \times \left(\prod_{i=1}^{2n-4}[q^2 z_i\bar w][q^2 \bar z_i \bar w][q \bar z_i w][q z_i w] \right)\MM_{1,\dots, 2n-4} (z_1,\dots,z_{2n-4})\ket{\downarrow\cdots \downarrow}.
  \end{multline}
  Upon half-specialisation and substitution into \eqref{eqn:ZZPhi}, we obtain the relation
  \begin{multline}
    \ZZ_n(w_1,\dots,w_{n-2},w,-q\bar w)=-[q^{1/2}]^2[w^2][q^2\bar w^2]F(w) \\
    \times \left(\prod_{i=1}^{n-2}[q^2 w_i\bar w][q^2 \bar w_i \bar w][q \bar w_i w][q w_i w]\right)^2\ZZ_{n-2}(w_1,\dots,w_{n-2}).
  \end{multline}
  The relation \eqref{eqn:ReductionY} for $n>2$ follows from combining this finding with \eqref{eqn:DefYY}.
\end{proof}

\subsubsection{Uniqueness}
We now relate the (rescaled) generalised sum of components $Y_N$ and the (rescaled) overlap $\YY_n$. The relation depends on the parity of $N$.

\begin{theorem}
  \label{thm:YandYY}
  For each $N\geqslant 0$ and $n$ as defined in \eqref{eqn:Defnnbar}, we have
  \begin{equation}
    Y_N(w_1,\dots,w_n)=\YY_n(w_1,\dots,w_n),
  \end{equation}
  where the parameters on the right-hand side are 
  \begin{equation}
   t=-\{\beta q^{1/2}\}/\{q^{1/2}\} , \quad b =
   \begin{cases}
     q, & \text{even }N,\\
     {\bar q}, & \text{odd }N.
   \end{cases}
   \label{eqn:paramBTQ}
  \end{equation}
\end{theorem}
\begin{proof}
  Define, for each $N\geqslant 0$ the function $X_N$ by $X_N(w_1,\dots, w_n) = \YY_n(w_1,\dots,w_n)$ with the parameters $t,b$ adjusted according to \eqref{eqn:paramBTQ}. We claim that the family of these functions obeys the properties \textit{(i)}-\textit{(iv)} of \cref{sec:Uniqueness}. 
  
  Indeed, the property \textit{(i)} straightforwardly holds by $\YY_0=\YY_1=1$. Property \textit{(ii)} follows from \cref{prop:YYBC}. For property \textit{(iii)}, we use \cref{prop:YYRed1} and the observation
  \begin{equation}
    (-1)^n\{b/q^{1/2}\} t = 
    (-1)^{n+1} \{\beta q^{1/2}\}/\{q^{1/2}\} \times
    \begin{cases}
     \{q^{1/2}\}, & \text{even }N,\\
     \{q^{3/2}\}, & \text{odd }N,
   \end{cases}
  \end{equation}
  for each $N\geqslant 2$. This evaluation can be written as
  \begin{equation}
     (-1)^n\{b/q^{1/2}\} t =
     (-1)^{n+1}\{q^{1/2}\beta\}c_N,
  \end{equation}
  where $c_N$ is defined in \cref{prop:ZBarReduction2}. Finally, to check the property \textit{(iv)}, we have to show that \cref{prop:YYRed2} holds with $F(w) = f_N(w)$, where $f_N$ is defined in \eqref{eqn:Deff}. To this end, we note that
  \begin{equation}
    \det \KK(\bar w) = t^2 - \left(\frac{\{q^{1/2}\bar w\}}{\{q^{1/2}\}}\right)^2 = \frac{[\beta w][\beta q\bar w]}{\{q^{1/2}\}^2}.
  \end{equation}
  Furthermore, the evaluation
  \begin{equation}
    [b\bar w][b\bar q w]= \begin{cases}
         [w][q\bar w], & \mathrm{even }\,N,\\
         [q w][q^2\bar w], & \mathrm{odd }\,N.
       \end{cases}
  \end{equation}
  yields, indeed, the desired equality.
  
  The theorem follows, therefore, from \cref{prop:Uniqueness}.
\end{proof}

\section{The main result}
\label{sec:MainResult}

We now combine the results of \cref{sec:GeneralisedSumComps,sec:6VTSASM} to express $S_N$ in terms of $A_{\mathrm{TS}}(2N+1;t,\tau)$, which constitutes our main result and proves Conjecture 5.5 of \cite{hagendorf:21}.\footnote{See also the erratum \cite{hagendorf:22_2}, which corrects a typographical error in the conjecture as stated in \cite{hagendorf:21}.} Both quantities are polynomials, but their relation involves the expression $(1+x(x-\tau))^{1/2}$. The relation is independent of the choice for the branch of the square root, but the same branch must be used consistently at every occurence of $(1+x(x-\tau))^{1/2}$.
\begin{theorem}
  \label{thm:MainTheorem}
  For each $N\geqslant 0$, we have
  \begin{equation}
    S_N = ((1+x(x-\tau))^{1/2})^nA_{\mathrm{TS}}(2N+1;t,\tau),
  \end{equation}
  where $t=(1+x)/(1+x(x-\tau))^{1/2}$.
\end{theorem}
\begin{proof}
  We combine \cref{prop:relationSZ} with \eqref{eqn:DefY}. Simplifying with the help of \eqref{eqn:Defnnbar}, we obtain
  \begin{equation}
    S_N = (-1)^{n'(n'-1)/2+n}([q]/[\beta])^n [q]^{-n(2n-1)}[q^{1/2}]^nY_N(1,\dots,1),
  \end{equation}
  where the parameters $x,\tau$ on the left-hand side are given by \eqref{eqn:ParametersXTau} in terms of $\beta,q$. Now, we apply \cref{thm:YandYY} and \eqref{eqn:DefYY}, which yields
   \begin{equation}
    S_N = ([q]/[\beta])^{n}[q]^{-n(2n-1)}(-1)^{n(n'-n)}[q^2]^{-n}\ZZ_n(1,\dots,1),
    \label{eqn:SIntermediate}
  \end{equation} 
  provided that \eqref{eqn:paramBTQ} holds. Evaluating \eqref{eqn:DefZZ} with $b=q$ for even $N$, and $b=\bar q$ for odd $N$, yields
  \begin{equation}
    (-1)^{n(n'-n)}[q^2]^{-n}\ZZ_n(1,\dots,1) = 
    \begin{cases}
      \mathbb Z_n^-(1,\dots,1), & \mathrm{even }\,N,\\
\mathbb Z_n^+(1,\dots,1),
		& \mathrm{odd}\,N.
    \end{cases}
  \end{equation}
    We substitute this evaluation into \eqref{eqn:SIntermediate} and use \cref{lem:TSASMGFZ} to obtain
  \begin{equation}
    S_N = ([q]/[\beta])^{n}A_{\mathrm{TS}}(2N+1;t,\tau),
    \label{eqn:SIntermediate2}
  \end{equation} 
  where $\tau$ and $t$ are defined in terms of $q$ and $\beta$ by \eqref{eqn:ParametersXTau} and \eqref{eqn:paramBTQ}, respectively. To eliminate $q$ and $\beta$ in favour of $\tau$ and $x$, we observe that $([q]/[\beta])^2 = 1+x(x-\tau)$ and hence
  \begin{equation}
    [q]/[\beta] = \eta(1+x(x-\tau))^{1/2},
    \label{eqn:QBeta}
  \end{equation}
  where $\eta^2=1$. Furthermore, a straightforward calculation yields $t= (1+x)[\beta]/[q]$. Using \eqref{eqn:QBeta}, we find
   \begin{equation}
    t = \eta\frac{1+x}{(1+x(x-\tau))^{1/2}}.
  \end{equation} 
Substituting the expressions for $[q]/[\beta]$ and $t$ into \eqref{eqn:SIntermediate2} gives
\begin{equation}
    S_N = \eta^n \left((1+x(x-\tau))^{1/2}\right)^n A_{\mathrm{TS}}\left(2N+1;\eta(1+x)/(1+x(x-\tau))^{1/2},\tau\right).
  \end{equation} 
  By \cref{cor:ParityATS}, the expression on the right-hand side is independent of $\eta$. We may therefore choose $\eta=1$, which concludes the proof.
\end{proof}
Setting $\tau =1$, we obtain the following result for the spin-chain Hamiltonian \eqref{eqn:XXZHamiltonian}, which accomplishes our main goal:
\begin{corollary}
  \label{cor:SumCompsSpecialEV}
  For each $N\geqslant 1$, the sum of components of the special eigenvector of \eqref{eqn:XXZHamiltonian} is
  \begin{equation}
    S_N = (1+x(x-1))^{n/2}A_{\mathrm{TS}}(2N+1;t,1),
  \end{equation}
  where $t = (1+x)/(1+x(x-1))^{1/2}$.
\end{corollary}

For $x=0$, where one of the boundary magnetic fields diverges, we have $t=1$. In this case, the sum of components yields
\begin{equation}
  S_N = A_{\mathrm{TS}}(2N+1;1,1)=|\mathrm{TSASM}(2N+1)|,
\end{equation}
which equals the number of TSASMs of order $2N+1$. Finding an explicit formula for $|\mathrm{TSASM}(2N+1)|$ for each $N\geqslant 0$ has so far remained an open problem in ASM enumeration. Several authors computed this number for small values of $N$ \cite{stanley:85,robbins:00,bousquet:95}. An immediate consequence of \cref{thm:MainTheorem} is the following multiple contour-integral formula:
\begin{corollary}
  \label{corr:TSASMContourIntegral}
  For each $N\geqslant 0$, we have
  \begin{multline}
   |\mathrm{TSASM}(2N+1)|=\oint\cdots \oint \prod_{k=1}^n\frac{\diff u_k}{2\pi \ii}\frac{(1+u_k + u_k^2)^{\epsilon}}{u_k^{n'+k-1}\left(1-\prod_{j=1}^ku_j\right)}\\
  \times \prod_{1\leqslant i \leqslant j \leqslant n}(1-u_iu_j)\prod_{1\leqslant i < j \leqslant n}(u_j-u_i)(1 + u_i + u_j)(1+ u_j + u_i u_j),
  \end{multline}
  where each integration contour is a simple positively-oriented curve around $0$.
\end{corollary}
As far as we are aware, this corollary provides the only known formula to date for the number of TSASMs of arbitrary order. Nonetheless, its practicality should be viewed critically. For instance, Behrend, Fischer, and Koutschan recently derived the leading term of the asymptotic series for $|\mathrm{TSASM}(2N+1)|$ as $N\to \infty$, despite the lack of an explicit formula  \cite{behrend:23}. Reproducing their results from the contour-integral formula poses a challenge.

We conclude this section with a discussion of the supersymmetric point $x=1$ of the spin chain, which originally motivated this study. For this purpose, we use the following result:
\begin{proposition}
   For each $N\geqslant 1$ and each $\tau$, $A_{\mathrm{TS}}(2N+1;1+\tau,\tau)=A_{\mathrm{TS}}(2N+3;1,\tau)$.
\end{proposition}
\begin{proof} 
  By \cite[Proposition 5.4]{hagendorf:21}, we have $\left.S_N\right|_{x=\tau} = \left.S_{N+1}\right|_{x=0}$. Using \cref{thm:MainTheorem}, the proposition follows.
   \end{proof}
Combining this proposition with \cref{thm:MainTheorem} and setting $\tau=1$, we immediately find:\begin{corollary}
  \label{corr:SumCompsSUSY}
  For $x=\tau=1$, we have $S_N = |\mathrm{TSASM}(2N+3)|$.
\end{corollary}
This confirms the observation made in \eqref{eqn:SumCompsSUSY} and generalises it to arbitrary $N$.


\begin{thebibliography}{10}

\bibitem{stroganov:01}
Yu.~G. {Stroganov},
\newblock {\em {The importance of being odd}},
\newblock J. Phys. A: Math. Gen. {\textbf{34}} (2001)   L179--L185.

\bibitem{razumov:01}
A.~V. {Razumov}, Yu.~G. {Stroganov},
\newblock {\em {Spin chains and combinatorics}},
\newblock J. Phys. A : Math. Gen. {\textbf{34}} (2001)   3185--3190.

\bibitem{razumov:01_2}
A.~V.~{Razumov} and Yu.~G.~{Stroganov},
\newblock {\em {Spin chains and combinatorics: twisted boundary conditions}},
\newblock J. Phys. A: Math. Gen. {\textbf{34}} (2001)   5335--5340.

\bibitem{batchelor:01}
M.~T. {Batchelor}, J.~{de Gier}, B.~{Nienhuis},
\newblock {\em {The quantum symmetric XXZ chain at {$\Delta=-1/2$},
  alternating-sign matrices and plane partitions}},
\newblock J. Phys. A: Math. Gen. {\textbf{34}} (2001)   L265--L270.

\bibitem{bressoudbook}
D.~Bressoud,
\newblock {\em {Proofs and confirmations: the story of the alternating sign
  matrix conjecture}},
\newblock Cambridge University Press, 1999.

\bibitem{difrancesco:06}
P.~{Di Francesco}, P.~{Zinn-Justin}, {J.-B.} {Zuber},
\newblock {\em {Sum rules for the ground states of the O(1) loop model on a
  cylinder and the XXZ spin chain}},
\newblock J. Stat. Mech. {\textbf{8}} (2006)  ~11.
\bibitem{razumov:07}
A.~V.~Razumov, Yu.~G.~Stroganov, P.~Zinn-Justin,
\newblock {\em {Polynomial solutions of qKZ equation and ground state of XXZ
  spin chain at $\Delta = -1/2$}},
\newblock J. Phys. A : Math. Gen. {\textbf{40}} (2007)   11827.

\bibitem{cantini:12_1}
L.~{Cantini},
\newblock {\em {Finite size emptiness formation probability of the XXZ spin
  chain at $\Delta=-1/2$}},
\newblock J. Phys. A: Math. Theor. {\textbf{45}} (2012)   135207.

\bibitem{morin:20}
A.~Morin-Duchesne, C.~Hagendorf, L.~Cantini,
\newblock {\em {Boundary emptiness formation probabilities in the six-vertex
  model at $\Delta=-1/2$}},
\newblock J. Phys. A: Math. Theor. {\textbf{53}} (2020)   255202.

\bibitem{hagendorf:21}
C.~Hagendorf, J. ~Li{\'{e}}nardy,
\newblock {\em {The open XXZ chain at ${\Delta} = -1/2$ and the boundary
  quantum Knizhnik-Zamolodchikov equations}},
\newblock J. Stat. Mech. (2021)   P013104.

\bibitem{hagendorf:22}
C.~Hagendorf, G.~Parez,
\newblock {\em {On the logarithmic bipartite fidelity of the open XXZ chain at
  $\Delta=-1/2$}},
\newblock SciPost Phys. {\textbf{12}} (2022)   199.

\bibitem{hagendorf:17}
C.~Hagendorf, J.~Li\'enardy,
\newblock {\em Open spin chains with dynamic lattice supersymmetry},
\newblock J. Phys. A: Math. Theor. {\textbf{50}} (2017)   185202.

\bibitem{pasquier:90}
V.~Pasquier and H.~Saleur,
\newblock {\em {Common structures between finite systems and conformal field
  theories through quantum groups}},
\newblock Nucl. Phys. B {\textbf{330}} (1990)   523 -- 556.

\bibitem{degier:04}
J.~{de Gier}, V.~{Rittenberg},
\newblock {\em {Refined Razumov Stroganov conjectures for open boundaries}},
\newblock J. Stat. Mech. {\textbf{9}} (2004)  ~9.

\bibitem{nichols:05}
A.~{Nichols}, V.~{Rittenberg}, J.~{de Gier},
\newblock {\em {One-boundary Temperley Lieb algebras in the XXZ and loop
  models}},
\newblock J. Stat. Mech. (2005)   P03003.

\bibitem{cherednik:92}
I.~Cherednik,
\newblock {\em {Quantum Knizhnik-Zamolodchikov equations and affine root
  systems}},
\newblock Comm. Math. Phys. {\textbf{150}} (1992)   109--136.

\bibitem{jimbo:95}
M.~Jimbo, R.~Kedem, H.~Konno, T.~Miwa, R.~Weston,
\newblock {\em Difference equations in spin chains with a boundary},
\newblock Nuclear Physics B {\textbf{448}} (1995)   429 -- 456.

\bibitem{difrancesco:07}
P.~{Di Francesco},
\newblock {\em {Open boundary quantum Knizhnik Zamolodchikov equation and the
  weighted enumeration of plane partitions with symmetries}},
\newblock J. Stat. Mech. {\textbf{1}} (2007)  ~24.

\bibitem{stokman:15}
J.~Stokman, B.~Vlaar,
\newblock {\em Koornwinder polynomials and the XXZ spin chain},
\newblock J. Approx. Theor. {\textbf{197}} (2015)   69--100.

\bibitem{reshetikhin:18}
N.~Reshetikhin, J.~Stokman, B.~Vlaar,
\newblock {\em Integral solutions to boundary quantum Knizhnik-Zamolodchikov
  equations},
\newblock Adv. Math. {\textbf{323}} (2018)   486 -- 528.

\bibitem{oeistsasm:24}
\newblock{The On-Line Encyclopedia of Integer Sequences. Sequence A005164.} 2024.

\bibitem{bousquet:95}
M.~Bousquet-M\'elou, L.~Habsieger,
\newblock {\em Sur les matrices \`a signes alternants},
\newblock Discr. Math. {\textbf{139}} (1995)   57 -- 72.

\bibitem{robbins:00}
D.~P.~{Robbins},
\newblock {Symmetry Classes of Alternating Sign Matrices},
\newblock arXiv:math.CO/0008045 (2000).

\bibitem{kuperberg:02}
G.~Kuperberg,
\newblock {\em {Symmetry Classes of Alternating-Sign Matrices under One Roof}},
\newblock Ann. Math. {\textbf{156}} (2002)   835--866.

\bibitem{behrend:23}
{R.~E.~Behrend}, {I.~Fischer}, {C.~Koutschan},
\newblock Diagonally symmetric alternating sign matrices,
\newblock arXiv:2309.08446 (2023).

\bibitem{hagendorf:22_2}
C.~Hagendorf and J.~Li\'enardy,
\newblock {\em {Erratum: The open XXZ chain at $\Delta=-1/2$ and the boundary
  quantum Knizhnik-Zamolodchikov equations (2021 J. Stat. Mech. 013104)}},
\newblock J. Stat. Mech. (2022)   059903.

\bibitem{sklyanin:88}
{E.~K.~Sklyanin},
\newblock {\em Boundary conditions for integrable quantum systems},
\newblock J. Phys. A: Math. Gen. {\textbf{21}} (1988)   2375.

\bibitem{zinnjustin:13}
P.~Zinn-Justin,
\newblock {\em {Sum Rule for the Eight-Vertex Model on Its Combinatorial
  Line}},
\newblock in K.~Iohara, S.~Morier-Genoud  and B.~R{\'e}my,
  editors, {\em Symmetries, Integrable Systems and Representations}, {\em {\em
  Volume}~40}, pages 599--637, Springer London, 2013.

\bibitem{hagendorf:16}
C.~{Hagendorf}, A.~{Morin-Duchesne},
\newblock {\em {Symmetry classes of alternating sign matrices in the
  nineteen-vertex model}},
\newblock J. Stat. Mech. (2016)   P053111.

\bibitem{brasseur:21}
S.~Brasseur and C.~Hagendorf,
\newblock {\em {Sum rules for the supersymmetric eight-vertex model}},
\newblock J. Stat. Mech. (2021)   P023102.

\bibitem{izergin:92}
A.~G.~{Izergin}, D.~A.~{Coker}, V.~E.~{Korepin},
\newblock {\em {Determinant formula for the six-vertex model}},
\newblock J. Phys. A: Math. Gen. {\textbf{25}} (1992)   4315--4334.

\bibitem{lienardy:tbp}
J.~Li{\'{e}}nardy, C.~Walmsley Hagendorf,
\newblock On the weighted enumeration of totally-symmetric alternating sign
  matrices,
\newblock to be published, 2025.

\bibitem{elkies:92}
N.~Elkies, G.~Kuperberg, M.~Larsen, J.~Propp,
\newblock {\em {Alternating-Sign Matrices and Domino Tilings (Part II)}},
\newblock J. Alg. Comb. {\textbf{1}} (1992)   219--234.

\bibitem{brualdi:13}
R.~A.~Brualdi, K.~P.~Kiernan, S.~A.~Meyer, M.~W.~Schroeder, 
\newblock{\em{Patterns of alternating sign matrices}},
\newblock Linear Algebra Appl. {\textbf{438}} (2013) 3967--3990.

\bibitem{okada:06}
S.~Okada,
\newblock {\em Enumeration of symmetry classes of alternating sign matrices and
  characters of classical groups},
\newblock J. Alg. Comb. {\textbf{23}} (2006) 43--69.

\bibitem{stanley:85}
R.~P.~Stanley,
\newblock {\em {A baker's dozen of conjectures concerning plane partitions}},
\newblock in {\em {Combinatoire \'enum\'erative}}, pages 285--293,
  Springer-Verlag New York, 1985.

\end{thebibliography}
\end{document}